\documentclass[lettersize,journal]{IEEEtran}
\usepackage{amsmath,amsfonts}
\usepackage{amssymb}
\usepackage{algorithmic}
\usepackage{algorithm}
\usepackage{array}
\usepackage[caption=false,font=normalsize,labelfont=sf,textfont=sf]{subfig}
\usepackage{textcomp}
\usepackage{stfloats}
\usepackage{url}
\usepackage{verbatim}
\usepackage{graphicx}
\usepackage{cite}
\usepackage{amsmath,amsthm}
\newtheorem{definition}{Definition}[section]

\newtheorem{lemma}{Lemma}[section]
\usepackage{graphicx}
\usepackage{algorithmic}
\usepackage{algorithm}
\usepackage{textcomp}
\usepackage{soul}
\usepackage{array}
\usepackage{amsmath}
\usepackage{cases}
\usepackage{tipa}
\usepackage{mathrsfs}
\usepackage{graphicx}
\usepackage{subfig}
\usepackage{bbding}
\usepackage{pifont}
\usepackage{multirow}
\usepackage{amsmath}

\usepackage{booktabs}
\usepackage{mdframed}
\usepackage{caption}
\usepackage{subcaption}
\usepackage{graphicx}
\usepackage{mathtools}

\begin{document}

\title{A2-DIDM: Privacy-preserving Accumulator-enabled Auditing for Decentralized Identity of DNN Model}

\author{Tianxiu~Xie,
Keke~Gai,~\IEEEmembership{Senior Member,~IEEE,}
Jing~Yu,~\IEEEmembership{Senior Member,~IEEE,}
Liehuang~Zhu,~\IEEEmembership{Senior Member,~IEEE}
\thanks{T. Xie, K. Gai and L. Zhu are with the School of Cyberspace Science and Technology, Beijing Institute of Technology, Beijing, China, 100081, and K. Gai is also with the School of AI, Beijing Institute of Technology, Beijing 100081, China, and Zhongguancun Academy, Haidian, Beijing, China.
(Emails: \{3120215672, gaikeke, liehuangz\}@bit.edu.cn.).}
\thanks{J. Yu is with the Key Laboratory of Ethnic Language Intelligent Analysis and Security Governance of MOE, Minzu University of China, and is also with the School of Information Engineering, Minzu University of China (e-mail: jing.yu@muc.edu.cn).}
\thanks{This work is partially supported by the National Key Research and Development Program of China (Grant No. 2021YFB2701300), the National Natural Science Foundation of China (Grant No.s U24B20146, 62372044), Beijing Municipal Science and Technology Commission Project (Z241100009124008), and Beijing Nova Program (Grant No. 20250484921).
}
\thanks{Keke Gai is the corresponding author (gaikeke@bit.edu.cn), J. Yu is a co-corresponding author (jing.yu@muc.edu.cn).}
}

\markboth{Journal of \LaTeX\ Class Files,~Vol.~XX, No.~X, Month~2024}%
{Shell \MakeLowercase{\textit{et al.}}: A Sample Article Using IEEEtran.cls for IEEE Journals}


\maketitle

\begin{abstract}
Recent booming development of Generative Artificial Intelligence (GenAI) has facilitated model commercialization to reinforce the model performance, including licensing or trading Deep Neural Network (DNN) models. 
However, DNN model trading may violate the benefit of the model owner due to unauthorized replications or misuse of the model.
Model identity auditing is a challenging issue in protecting DNN model ownership, and verifying the integrity and ownership of models is one of the critical obstacles. 
In this paper, we focus on the above issue and propose an \underline{A}ccumulator-enabled \underline{A}uditing for \underline{D}ecentralized \underline{Id}entity of DNN \underline{M}odel (A2-DIDM) that utilizes blockchain and zero-knowledge techniques to protect data and function privacy while ensuring the lightweight on-chain ownership verification. 
The proposed model presents a scheme of identity records via configuring model weight checkpoints with zero-knowledge proofs, which incorporates predicates to capture incremental state changes in model weight checkpoints.  
Our scheme ensures both computational integrity and programmability in DNN training process so that the uniqueness of the weight checkpoint sequence in a DNN model is preserved. 
A2-DIDM also addresses privacy protections in decentralized identity.
We systematically analyze the security and robustness of our proposed model and further evaluate the effectiveness and usability of auditing DNN model identities. 
The code is available at https://github.com/xtx123456/A2-DIDM.git.

\end{abstract}


\section{Introduction}\label{intro}

\IEEEPARstart{A}{s} {\em Generative Artificial Intelligence} (GenAI) continues to drive the growing demand for data and model trading, such as licensing DNN models on third-party commercial platforms. 
{\em Model-as-a-Service} (MaaS) has emerged as a representative commercialization paradigm and has been applied to natural language processing \cite{min2023recent}, image recognition \cite{zhang2023vitaev2}, and intelligent decision-making \cite{yang2019unremarkable}. 
However, once a model is traded or its parameters are illegally exposed, it may be redistributed or misused by unauthorized parties, including competitors and adversaries. 
Our research addresses the model identity issue and consider model identity auditing a promising direction for protecting model {\em Intellectual Property} (IP) and enabling secure model commercialization \cite{peng2023intellectual, chen2023deepjudge, li2019prove, lao2022identification}. 
Nevertheless, existing DNN IP protection schemes are still inadequate for model identity auditing. 
An ideal solution should verify ownership, guarantee training integrity, and validate training parameters to prevent unauthorized data theft or redistribution.


Existing schemes can be typically categorized into three types\cite{dong2023rai2}.
The first type identifies DNN models through watermarking or fingerprinting. Watermarking embeds secret marks into model weights, activation functions, or architectures, but may degrade model performance \cite{zhang2020model, szyller2021dawn, zhao2023protecting}. Fingerprinting generates unique identifiers from model features or metadata, yet such identifiers can be invalidated by model distillation or extraction attacks, and even forged by adversaries to claim ownership \cite{peng2022fingerprinting, yu2021artificial}.
Second, training dataset inference determines whether a model is stolen by leveraging the model owner’s private knowledge of the training dataset and decision boundaries \cite{maini2021dataset, dziedzic2022difficulty, dziedzic2022dataset}. 
However, such schemes may expose sensitive data checkpoints, which significantly limits their practical adoption.
Finally, computation-based model ownership verification schemes replay the training history to validate the computational integrity of the claimed model owner throughout training \cite{jia2021proof,liu2023provenance,choi2023tools}.
These methods remain vulnerable to “adversarial samples” attacks and still face the risk of proof forgery \cite{zhang2022adversarial, fang2023proof}.
Therefore, we argue that the key challenge in model ownership verification is to design an efficient and lightweight model identity auditing scheme that does not disclose the model owner’s private knowledge during training.
Moreover, due to the modifications made to the model structure during inference auditing of DNN models\cite{abbaszadeh2024zero, garg2023experimenting}, we prioritize the integrity auditing of the model training process over the inference process.

By leveraging blockchain, a decentralized identity for DNN models can serve as a trusted identity auditing mechanism, providing a new perspective for model ownership verification and commercialization \cite{2019EverSSDI,10229502}. 
An effective \underline{D}ecentralized \underline{Id}entity of DNN \underline{M}odel (DIDM) should satisfy two properties: (i) \textbf{training integrity}, namely that model owners have genuinely expended computational resources for training; and (ii) \textbf{execution correctness}, namely that the training tasks are carried out correctly.
To address the above issues, inspired by decentralized verifiable computation \cite{xiong2023verizexe, bowe2020zexe}, we propose a novel \underline{A}ccumulator-enabled \underline{A}uditing for \underline{D}ecentralized \underline{Id}entity of DNN \underline{M}odel (A2-DIDM). 
A2-DIDM integrates lightweight proving techniques with zero-knowledge mechanisms to enable auditable and verifiable model identity while preserving the privacy of training information. 
To support lightweight operations, A2-DIDM introduces the {\em Identity Record} (IR) as the minimal unit of decentralized identity management. 
Each IR captures the incremental update between adjacent model weight checkpoints during DNN training,
An ordered sequence of IRs forms the unique training history of the final model, thereby ensuring training integrity. 
By recording IRs (or their commitments) on the blockchain, A2-DIDM enables effective auditing of the optimization process performed by the model owner to obtain the converged model. 
Moreover, the policies and logic governing checkpoint updates are formalized as predicates that regulate IR creation and execution. 
Satisfiability of these predicates indicates that the IRs originate from genuine training over private data, thereby ensuring execution correctness. 
Thus, A2-DIDM allows the model owner to train off-chain and subsequently validate the DIDM on-chain using accumulator-based proofs.

Furthermore, to achieve constant on-chain verification cost without disclosing checkpoints or predicate details, A2-DIDM adopts an accumulator-based zkSNARK proving system. 
Rather than revealing IRs directly, the model owner generates a zkSNARK proof asserting that all predicates in the DIDM instance are satisfied, enabling zero-knowledge verification of both training integrity and execution correctness. 
Concretely, the model owner first generates internal proofs ($\pi_{in}$) for IRs to prove that the predicates associated with the corresponding checkpoints are satisfied.
Then it computes an external proof ($\pi_{out}$) to verify the correctness of these internal proofs while ensuring unlinkability between predicates and IRs. 
To reduce the cost of the external proof, the accumulator defers the expensive pairing operations of the external zkSNARK to on-chain computation, while preserving lightweight verification for both the accumulator and the external proof. 
Neither $\pi_{in}$ nor $\pi_{out}$ reveals any information about the DNN training data or predicate/function implementations. 
Finally, the model owner submits the IR commitments and the DIDM zkSNARK proof ($\pi_{out}$) to the blockchain verifier, which verifies $\pi_{out}$ on-chain at constant cost and records the corresponding IR commitments upon acceptance. 
The blockchain timestamps of IRs further mitigate replay attacks, such as preemptive ownership claims by adversaries.

Our contributions in this paper are summarized as follows:
\begin{enumerate}
    \item We propose a novel accumulator-based model ownership verification scheme that combines verifiable computation and decentralized identity. A2-DIDM models the process of weight checkpoint changes in a DNN model by attaching specific predicates to IRs. It realizes auditable and programmable DNN training processes for the model owner, while ensuring lightweight proof operations.
    \item The accumulator guarantees the computational integrity without exposing any dataset or weight checkpoint. It also hides all function states within the predicates, including those triggered by each IR. Specifically, A2-DIDM achieves data privacy and function privacy for model owners without altering the internal structure of the model or the training optimization algorithm.
    \item In A2-DIDM, model owners perform DNN training computations off-chain and validate their identity on-chain. On-chain verifiers only need to verify the short IR validity proof instead of re-executing the DNN training computations. Regardless of the complexity of off-chain computations, the on-chain verification time cost remains constant.
\end{enumerate}


The organization of this work follows the order below.
Related work is given in Section \ref{sec:rw}.
Section \ref{sec:pre} provides the preliminaries.
Sections \ref{sec:mod} and \ref{sec:ana} present the design and security analysis of the proposed model, respectively. 
Sections \ref{sec:exp} provide experiment evaluations with findings and we draw our conclusions in Section \ref{sec:con}.
\section{Related Work} \label{sec:rw}

\textbf{Deep Learning IP Protection.} 
Existing studies show that once a model is released under MaaS, adversaries may replicate or steal it through remote API access, threatening the legitimate owner’s interests \cite{zheng2022dnn, ijcai2025p831}. 
This has motivated extensive research on IP protection and ownership verification. 
Model watermarking is a representative approach. SSLGuard \cite{cong2022sslguard} targets pre-trained encoders in self-supervised learning, while Wen {\em et al.} \cite{wen2023function} enhance watermark robustness by generating watermark triggers from original training samples. 
In distributed machine learning, watermarks have also been embedded into global models for ownership verification \cite{liu2025tracemop,chen2023fedright,gai2025mfl}.
For example, MFL-Owner injects client-distinguishable watermarks for multimodal federated learning without degrading task performance \cite{gai2025mfl}. 
However, model watermarking generally faces a trade-off between practicality and robustness, and often requires modifications to the training algorithm or model structure. 
Moreover, model stealing attacks can remove or invalidate watermarks, thereby infringing model copyrights \cite{jia2021entangled, sun2023denet}. 
To address these limitations, Guan {\em et al.} \cite{guan2022you} exploit sample correlation differences to detect model stealing, and Li {\em et al.} \cite{li2022defending} identify victim models through exogenous feature embedding and meta-classifier training. 
More recently, the notion of model identity has been discussed in RAI2 \cite{dong2023rai2} and SoK \cite{chandrasekaran2021sok}. 
Nevertheless, the above model identity auditing schemes all compromise the model owner’s privacy. 
By contrast, our approach preserves both data privacy and function privacy via two-step incremental verifiable computation, while ensuring computational integrity and execution correctness.

\textbf{Verifiable Computing in Cryptography.}
Verifiable computation enables the model owner to outsource a computing task to an untrusted third-party executor while obtaining a proof that the returned result is correct \cite{zhang2023survey, li2023martfl,du2024towards}.
For example, 
Niu {\em et al.} \cite{9247447} developed a verifiable and privacy-preserving machine learning method that maintained function privacy through oblivious evaluation and batch result verification.
Moreover, verifiable computation could also be used for verifiable federated learning.
Guo {\em et al.} \cite{9285303} propose \textit{VeriFL}, a verifiable aggregation protocol with dimension-independent communication cost and bounded computational overhead.
Additionally, some research focuses on privacy-preserving DML training \cite{liu2021leia,liu2021towards,10946247,garg2023experimenting,abbaszadeh2024zero}. 
Leia \cite{liu2021leia} turns neural network inference into a two–non-colluding-edge secure protocol by rewriting the model as BNN and using lightweight secret-sharing-friendly operators.
The zkPoT scheme\cite{garg2023experimenting} improves verifiable computation schemes using MPC-in-the-head and realize zero-knowledge proof for model training. 
Nevertheless, these zero-knowledge DML verification schemes do not embed the original training algorithm losslessly. They approximate non-arithmetic operations for arithmetization, which alters training and typically proves only one epoch. 
By contrast, A2-DIDM commits to the training history without changing the DNN model or its training algorithms.

\section{Preliminaries}\label{sec:pre}


\textbf{Notation.}
In this paper, $\mathbf{W}$ $=$ $[W_{0},W_{1},W_{2},...,W_{P}]$
represents the weight checkpoint sequences of DNN, where $W_{i}(i \in [0,P])$ represents the weight parameters at the $i$-th epoch.
Additionally, for an $T$-layer DNN model, the checkpoint at the $i$-th epoch is denoted as $W_{i}$ $=$
[$W^{(0)}_{i}, W^{(1)}_{i}, ..., W^{(T)}_{i}$], where the entry $W^{(t)}_{i}(t \in [0,T])$ represents the $t$-th layer weight parameters of DNN.
Furthermore, $\mathbb{G}$ denotes a finite field and the set of integers.

\textbf{Collision-resistant Hash Function.}
Collision-resistance and one-wayness are the key properties of hash functions.
A typical hash function $\mathsf{CRH}$ involves two algorithms as follows: 
\begin{itemize}
    \item $\mathtt{pp_{h}}$ $\gets$ $\mathsf{CRH.Gen}(1^{\lambda})$: Inputs the security parameter ($\lambda$) and outputs the public parameters ($\mathtt{pp_{h}}$) for collision-resistant hash function.
    \item $\mathtt{h}(m)$ $\gets$ $\mathsf{CRH.Val}(\mathtt{pp_H},m)$: On input public parameters ($\mathtt{pp_{H}}$) and message ($m$), output a short hash value $\mathtt{h}(m)$.
\end{itemize}

\textbf{Commitment Scheme.}
A commitment scheme is a cryptographic primitive that enables a sender (\textbf{prover}) to commit to a message while keeping it hidden, and later open the commitment to reveal the message.
A commitment scheme $\mathsf{CS}$ is typically defined by three Probabilistic Polynomial-Time (PPT) algorithms,
$\mathsf{CS}=(\mathsf{Gen},\mathsf{Commit},\mathsf{Open})$:
\begin{itemize}
  \item $\mathtt{pp_{C}} \leftarrow \mathsf{CS.Gen}(1^{\lambda})$: on input the security parameter $\lambda$, output the public parameters $\mathtt{pp_{C}}$.
  \item $\mathtt{com} \leftarrow \mathsf{CS.Commit}(\mathtt{pp_{C}}, m; r)$: on input a message $m$ and randomness $r$, output a commitment $\mathtt{com}$.
  \item $b \leftarrow \mathsf{CS.Open}(\mathtt{pp_{C}}, \mathtt{com}, m; r)$: output $b\in\{0,1\}$ indicating whether $\mathtt{com}$ is a valid commitment to $m$ under randomness $r$.
\end{itemize}

A secure commitment scheme satisfies:
\begin{itemize}
  \item \textbf{Binding.} It is computationally infeasible for any PPT adversary to find
  $(m,r)\neq(m',r')$ such that
  $\mathsf{Commit}(\mathtt{pp_{C}}, m; r)=\mathsf{Commit}(\mathtt{pp_{C}}, m'; r')$.
  \item \textbf{Hiding.} For any PPT distinguisher and any equal-length messages $m_0,m_1$,
  the distributions $\mathsf{Commit}(\mathtt{pp_{C}}, m_0; r)$ and
  $\mathsf{Commit}(\mathtt{pp_{C}}, m_1; r)$ are computationally indistinguishable,
  where $r \overset{\$}{\leftarrow} R_{\mathtt{pp_{C}}}$.
\end{itemize}

\textbf{Non-interactive Zero-Knowledge (NIZK) Arguments.}
Non-interactive zero-knowledge (NIZK) arguments allow a \textbf{prover} to convince a \textbf{verifier} that an NP statement is true using a single proof, without interaction and without revealing any information beyond the validity of the statement.
Let $\mathbf{R}\subseteq \mathcal{S}\times\mathcal{W}$ be a binary relation, where $S\in\mathcal{S}$ is an instance and $J\in\mathcal{W}$ is a witness. A NIZK argument system is defined by four PPT algorithms,
$\mathsf{NIZK}=(\mathsf{Gen},\mathsf{KeyGen},\mathsf{Prove},\mathsf{Verify})$:
\begin{itemize}
  \item $\mathtt{pp_{ZK}} \leftarrow \mathsf{NIZK.Gen}(1^{\lambda},\mathbf{R})$:
  on input the security parameter $\lambda$ and relation $\mathbf{R}$, output the public parameters $\mathtt{pp_{ZK}}$.

  \item $(\mathtt{pk},\mathtt{vk}) \leftarrow \mathsf{NIZK.KeyGen}(\mathtt{pp_{ZK}})$:
  on input $\mathtt{pp_{ZK}}$, output a proving key $\mathtt{pk}$ and a verification key $\mathtt{vk}$.

  \item $\pi \leftarrow \mathsf{NIZK.Prove}(\mathtt{pp_{ZK}},\mathtt{pk}, S, J)$:
  on input an instance-witness pair $(S,J)\in \mathbf{R}$, output a proof $\pi$.

  \item $b \leftarrow \mathsf{NIZK.Verify}(\mathtt{pp_{ZK}},\mathtt{vk}, S, \pi)$:
  on input the instance $S$, verification key $\mathtt{vk}$, and proof $\pi$, output $b\in\{0,1\}$ indicating acceptance if $b=1$ and rejection otherwise.
\end{itemize}

A NIZK argument system typically satisfies the following properties:
\begin{itemize}
  \item \textbf{Completeness.} For any $(S,J)\in\mathbf{R}$, an honest \textbf{prover} can generate a proof that an honest \textbf{verifier} accepts.

  \item \textbf{Knowledge soundness.} No efficient PPT adversary can convince the \textbf{verifier} to accept a false statement except with negligible probability.

  \item \textbf{Zero knowledge.} The proof reveals no information beyond the validity of the statement.
\end{itemize}

\textbf{Principal Component Analysis (PCA).}
PCA is a linear dimensionality reduction technique that projects features onto an orthogonal basis ordered by variance. In the context of parameter independence, uniform feature contributions and the absence of dominant correlations disperse variance across multiple components, preventing the first principal component from dominating.

\section{Definitions and Models}\label{sec:mod}

A2-DIDM verifies both the integrity and correctness of the final DIDM. Unlike methods that require modifications to the training architecture or optimization procedure, it only requires the model owner ($\mathcal{P}$) to generate IRs for successive checkpoints ($\mathbf{W}$) and maintain the corresponding zkSNARK proofs. During training, $\mathcal{P}$ keeps the private dataset ($\mathbf{D}\mathcal{P}$) and checkpoint sequence ($\mathbf{W}$) off-chain without revealing any parameters. Once all predicate checks for the incremental updates in an IR are satisfied, $\mathcal{P}$ generates the corresponding external proof ($\pi_{out}$). When claiming ownership of the converged model ($f_{W_{P}}$), the on-chain verifier ($\mathcal{V}$) only checks the commitments and proofs ($\pi_{out}$) in the submitted IRs, thereby confirming that the IRs are from genuine DNN training rather than a forged history, while preserving the privacy of datasets and model parameters.

\subsection{Threat Model for A2-DIDM}



A2-DIDM involves three parties: the legitimate model owner ($\mathcal{P}$), an on-chain verifier ($\mathcal{V}$), and an adaptive PPT adversary ($\mathcal{A}$). 
$\mathcal{P}$ trains the DNN model and generates IRs together with the corresponding zkSNARK proofs; 
$\mathcal{V}$ verifies the IRs and proofs using the public verification procedures and outputs accept/reject; 
$\mathcal{A}$ attempts to pass verification without conducting the genuine training in order to illicitly claim model ownership.
In addition, we assume that $\mathcal{V}$ only performs the protocol-defined verification operations (i.e., checking the commitments and zkSNARK proofs), without retraining the full DNN or reproducing the training history. 
Meanwhile, $\mathcal{P}$ preserves the privacy of training dataset ($\mathbf{D}_\mathcal{P}$) and the checkpoint sequence ($\mathbf{W}=[{W_0,\ldots,W_P}]$). 
Only the IR commitments and the zkSNARK proofs are published publicly on the blockchain. 
Additionally, we assume standard security properties for the underlying cryptographic primitives, specifically commitment binding and hiding, as well as NIZK completeness, knowledge soundness, and zero-knowledge. 
Finally, we assume that the blockchain is a persistent, immutable distributed ledger resilient to infrastructure threats such as 51\% attacks or eclipse attacks.

\textit{Adversary's Objective:} 
The primary objective of $\mathcal{A}$ is to deceive $\mathcal{V}$ into accepting a fraudulent claim of having executed the prescribed training process on the dataset.
Specifically, $\mathcal{A}$ fabricates a set of IRs and proofs leading to the target model ($f_{W_P}$), thereby simulating a plausible training history and zkSNARK proofs without actual execution.

\textit{Adversary's Knowledge:}
$\mathcal{A}$ possesses white-box access to the converged DNN model ($f_{W_P}$), granting full knowledge of its architecture and final weight parameters ($W_P$). 
Additionally, the adversary $\mathcal{A}$ may have access to an auxiliary dataset ($\mathbf{D}_\mathcal{A}$) whose distribution is similar to that of the private training dataset ($\mathbf{D}_\mathcal{P}$). 
We note that a subset ($\mathbf{D}_\mathcal{A}^{l}$) of samples in $\mathbf{D}_\mathcal{A}$ are correctly labeled.
However, the labeled subset ($\mathbf{D}_\mathcal{A}^{l}$) is significantly smaller in size (i.e., $|\mathbf{D}_\mathcal{A}^{l}|\ll|\mathbf{D}_\mathcal{P}|$), and $\mathcal{A}$ is permitted to perform arbitrary computations upon it. 
While $\mathcal{A}$ is fully aware of the protocol specification, IR format, and all public parameters, it does not possess the original training dataset ($\mathbf{D}_\mathcal{P}$), the genuine checkpoints (${W_i}$), or the secret randomness and witnesses used by $\mathcal{P}$ to generate valid proofs.

\textit{Adversary's Capability:}
We consider $\mathcal{A}$'s ability to effectively deceive $\mathcal{V}$ from the following four aspects:
\begin{enumerate}
    \item[\textbf{M1.}] 
    Without running genuine training, $\mathcal{A}$ can fabricate intermediate weight checkpoints by gradually interpolating between a randomly initialized parameter vector and the target final parameters ($W_P$). 
    This process generates a smooth, ostensibly continuous sequence terminating at $W_P$ that bypasses actual optimizer dynamics, aiming to satisfy consistency requirements during IR verification.

    \item[\textbf{M2.}] 
    Knowing $W_P$, instead of using random initialization, $\mathcal{A}$ may set the initial parameters to $\widetilde{W}_0$ that are close to $W_P$ under a chosen metric (e.g., small $\lVert \widetilde{W}_0 - W_P\rVert$). 
    This reduces the apparent number of training steps needed to reach the final model and makes it easier to fabricate a smooth-looking sequence of weight checkpoints.
    
    \item[\textbf{M3.}] 
    Given  white-box access to final DNN model ($f_{W_P}$), $\mathcal{A}$ can treat $f_{W_P}$ as a teacher model and perform knowledge distillation or imitation learning on an auxiliary dataset ($\mathbf{D}_\mathcal{A}$) to construct a forged parameter sequence ($\widetilde{\mathbf{W}}=[{\widetilde{W}_0,\ldots,\widetilde{W}_P}]$) such that $\widetilde{W}_P = W_P$. 
    The purpose is to produce weight checkpoints that appear to reflect gradual convergence, even though the updates are not required to be genuine optimizer steps on the private training dataset ($\mathbf{D}_\mathcal{P}$).

    \item[\textbf{M4.}]
    $\mathcal{A}$ may attempt protocol-level attacks, including forging NIZK proofs, forging zkSNARK proofs, or producing verifiable IRs without satisfying the required predicates. 
    In the security analysis, we rule out these attacks by relying on the security of the underlying primitives, such as commitment binding and NIZK zero-knowledge and knowledge soundness.

\end{enumerate}

\subsection{Distributed Identity Construction of DNN Model}

In our proposed A2-DIDM, the authorized model owners with access to blockchain can execute DNN training computation off-chain. 
Moreover, a series of IR commitments are recorded in on-chain transactions only if all associated predicates are satisfied.

\begin{figure}[t]
    \centering
    \includegraphics[width=0.9\columnwidth]{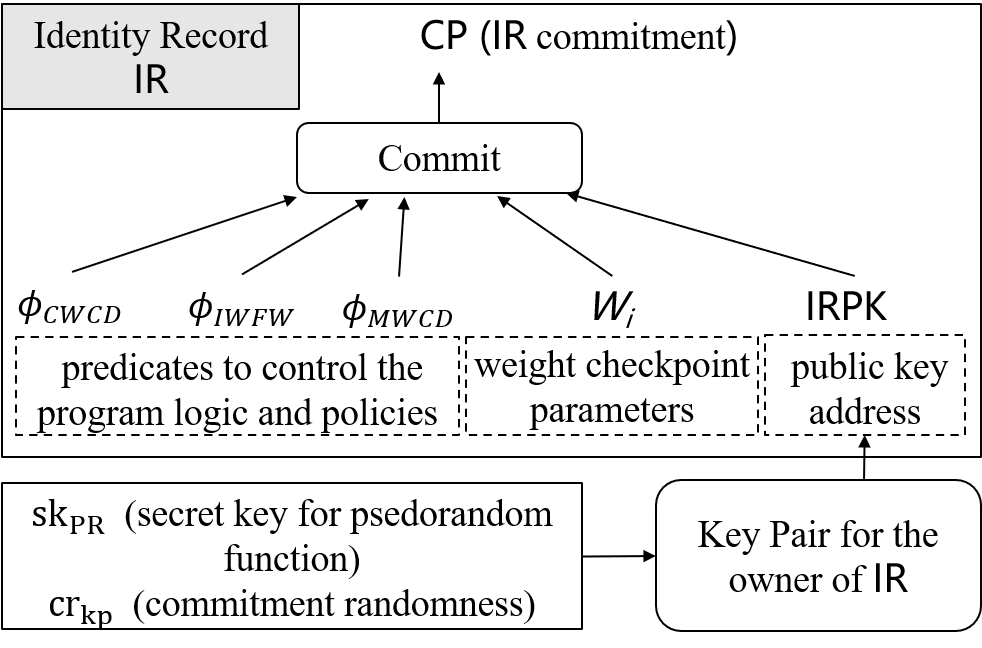}
    \caption{The construction of Identity Record for DIDM.}
    \label{fig1}
\end{figure}

\subsubsection{Identity Record and Transaction Design}

The definition of IR is shown as Definition \ref{def1}.
We note that any sequence of weight checkpoint states is indexed by the training step $i \in [0,P]$. 
All predicates have access to auxiliary information ($\mathtt{AUX}$), which is not publicly available on the blockchain. 

\begin{definition}
    (Identity Record). 
    Given the private dataset $\mathbf{D}_\mathcal{P}$, for a model owner $\mathcal{P}$, a valid Identity Record (IR) is defined as $\mathtt{IR}$ $=$ $(\mathtt{CP}$, $\mathtt{IRPK}$, $W_i$, $\Phi_{\mathtt{CWCD}}$, $\Phi_{\mathtt{IWFW}}$, $\Phi_{\mathtt{MWCD}}$, $\mathtt{AUX}$$)$. Specifically, $\mathtt{IR}$ has the following attributes: (a) a commitment $\mathtt{CP}$ for a weight checkpoint $W_{i}$ from the checkpoint sequence $\mathbf{W}$, combines all the remaining attributes of $\mathtt{IR}$ without information disclosure; 
    (b) $\mathtt{IRPK}$ denotes the public key address which represents the owner of the $\mathtt{IR}$; 
    (c) $W_i$ denotes the weight checkpoint parameters at the current training epoch; 
    (d) $\Phi_{\mathtt{CWCD}}$ denotes the predicate for the continuity of weight checkpoint distribution; 
    (e) $\Phi_{\mathtt{IWFW}}$ denotes the predicate for the distance between initial weight and final weight; 
    (f) $\Phi_{\mathtt{MWCD}}$ denotes the predicate for the monotonicity of weight checkpoint distribution; 
    (g) the cryptography auxiliary information, denoted by $\mathtt{AUX}$, includes accumulator parameters and commitment randomness, etc. 
\end{definition}\label{def1}


The life-cycle of DIDM starts with the model owner that executes DNN training computing off-chain by generating new IRs.
The construction of IR and transaction are directly shown in Fig.s \ref{fig1} and \ref{fig2}, respectively.
Fig. \ref{fig2} illustrates the construction of valid transactions on the blockchain, which includes a ledger digest, new IR commitment ($\mathtt{CP}$), and zkSNARK proof ($\pi_{out}$). 
For a valid transaction, the commitments of all IRs undergo multiple hash operations, and the resulting Merkle tree root serves as the ledger digest that is stored. 
The commitments involve the weight information in the IRs, ensuring data privacy during IR execution. 
In the NIZK arguments, we formalize the predicate satisfiability of DID audit as a binary relation ($\mathbf{R}$) that captures both rightful ownership of the IR (i.e., a valid opening of the IR commitment) and correctness of predicate proof. 
Specifically, the process of proving relation ($\mathbf{R}$) is modeled as a two-step incremental verifiable computation, consisting of an internal zkSNARK proof ($\pi_{in}$) and an external zkSNARK proof ($\pi_{out}$).
The internal proof ($\pi_{in}$) attests to the valid IR commitment opening and the satisfaction of all predicates ($\Phi_{\mathtt{CWCD}}$, $\Phi_{\mathtt{IWFW}}$, $\Phi_{\mathtt{MWCD}}$) under the weight checkpoint ($\mathbf{W}$) of IR.
To achieve function privacy, the internal proof ($\pi_{in}$) is not directly published on the blockchain. 
Instead, an external zkSNARK proof ($\pi_{out}$) is generated to ensure the correctness of the internal proof ($\pi_{in}$) regarding the predicates.
Accordingly, the external proof ($\pi_{out}$) asserts that all predicates associated with the IRs are satisfied, while preserving function privacy throughout the zero-knowledge proof.
Finally, an on-chain zkSNARK \textbf{verifier} validates the external proof ($\pi_{out}$). 
This process confirms the validity of the computation without revealing the actual predicates involved in the IR.

\begin{figure}[t]
    \centering
    \includegraphics[width=0.9\columnwidth]{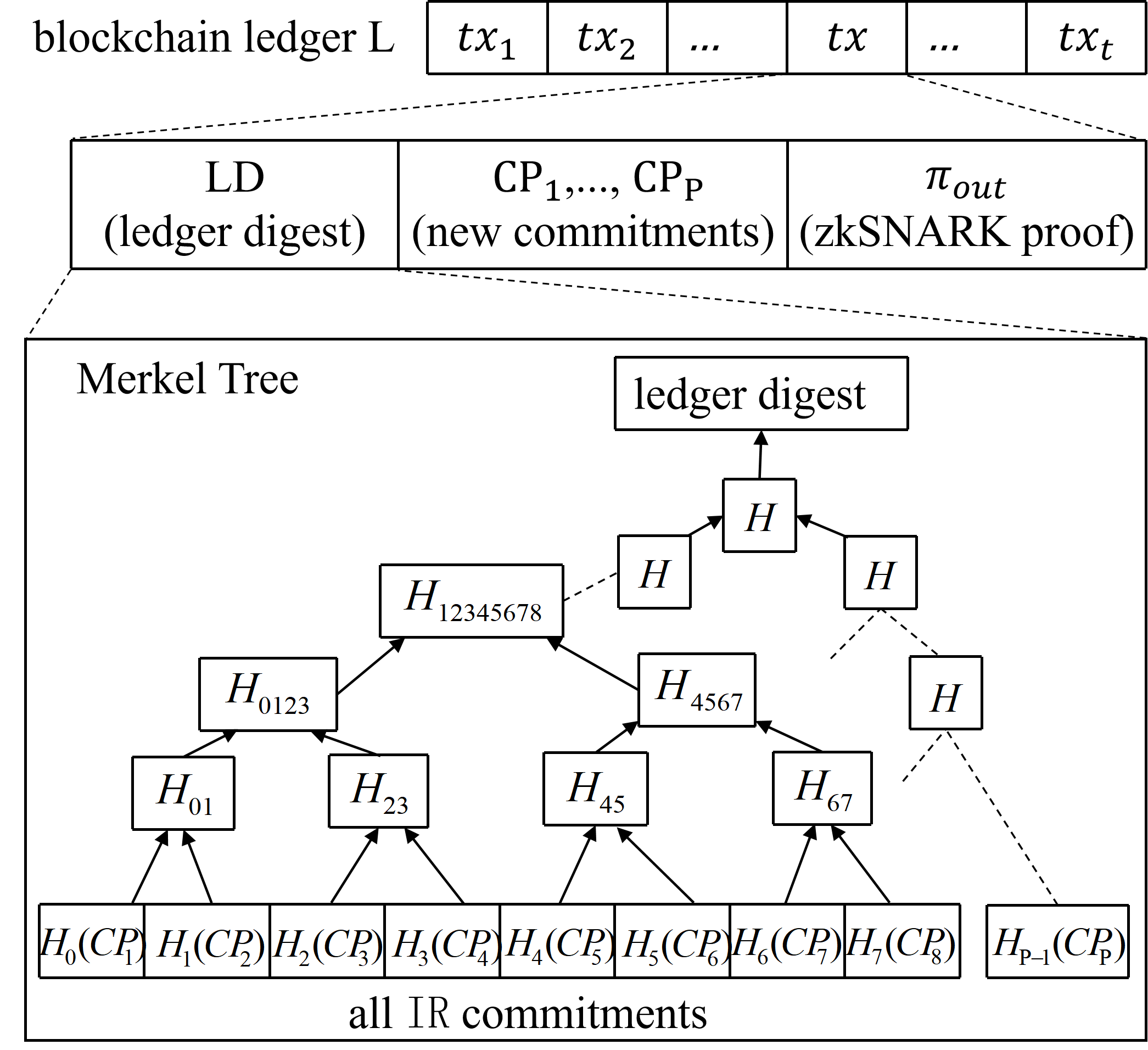}
    \caption{The construction of valid transactions on blockchain.}
    \label{fig2}
\end{figure}


\subsubsection{Predicate Design}

The predicates in IRs govern the incremental updates between consecutive weight checkpoints.
To ensure correct DIDM auditing, we require verifiable uniqueness of the checkpoint sequence ($\mathbf{W}$) for the final DNN model ($f_{W_{P}}$). 
Specifically, with an auxiliary dataset ($\mathbf{D}_\mathcal{A}$), $\mathcal{A}$ is unable to generate an alternative IR sequence which is still accepted by $\mathcal{V}$. 
Based on the threat model and the two effectiveness properties of DIDM, A2-DIDM should satisfy the following predicates to guarantee IR-sequence uniqueness.


\textbf{$\Phi_{\mathtt{CWCD}}$ (Continuity of Weight Checkpoint Distance).} 
Since weight parameters are updated incrementally, intermediate checkpoints evolve smoothly during training. By the convergence property of SGD, the optimization trajectory gradually contracts toward the final checkpoint $W_{P}$. We define the normalized distance between checkpoints by Eq. (\ref{eq:dl}):
\begin{equation}\label{eq:dl}
\mathsf{DL}(W_i,W_j) = \frac{\sqrt{\sum^{T}{t=0} \Vert W_i^{(t)}- W_j^{(t)} \Vert ^{2}{F}}}{Total_{\mathbf{W}}}.
\end{equation}
Then, since a randomly initialized checkpoint $W_{\mathsf{ran}}$ is generally farther from $W_{P}$ than the true initial checkpoint $W_{0}$, predicate $\Phi_{\mathtt{CWCD}}$ is defined as
\begin{equation}
\mathsf{DL}(W_{\mathsf{ran}},W_{P}) \gg \mathsf{DL}(W_{0},W_{P}).
\label{eq:1}
\end{equation}
Specifically, we sample a set of independent random initial weight checkpoints from the same initialization scheme as $W_0$ , and calculate the distance from these random initial checkpoints to $W_P$ respectively. Let $\mathrm{DisMean}$ and $\mathrm{DisStd}$ denote the mean and standard deviation of this set of distances, we satisfy $\Phi_{\mathtt{CWCD}}$ if
$
\mathsf{DL}(W_{0},W_{P}) \le \mathrm{DisMean} - \epsilon \cdot \mathrm{DisStd},
$
where $\epsilon$ is a safety margin.
Moreover, $\mathsf{DL}(W_{\mathsf{ran}}, W_{P})$ computes the overall distance across all corresponding weights in both DNN models.
Since weights in the same layers are identically distributed, $\mathsf{DL}(W_{\mathsf{ran}}, W_{P})$ is followed a normal distribution based on the central limit theorem.
Its standard deviation ($\mathrm{DisStd}$) scales inversely with the square root of the total weight count ($Total_{\mathbf{W}}$). 
Consequently, we set $\epsilon=5$ to achieve a 99.9999\% confidence level in our implementation. 
This indicates that for any random initialization model ($W_{\mathsf{ran}}$), the probability of obtaining $\mathsf{DL}(W_{\mathsf{ran}},W_{P}) \le \mathrm{DisMean} - 5 \times \mathrm{DisStd}$ is approximately $10^{-7}$.

\textbf{$\Phi_{\mathtt{IWFW}}$ (Small Distance from Initial Weight to Final Weight).} 
lthough an adversary may manipulate the forged initial checkpoint $W_0'$ to be close to the converged checkpoint $W_P$, random initialization can still be identified through the distributional property and mutual independence of the initial weights. 
Inspired by PoT \cite{liu2023provenance}, we verify whether the weights in each model layer are drawn from a two-component GMM distribution. 
For the $t$-th layer of the initial checkpoint $W_0^{(t)}$, let $\mathsf{WD}(W_0^{(t)})$ and $\mathsf{GD}(W_0^{(t)})$ denote its weight distribution and the GMM distribution, respectively. 
In addition, $\Phi_{\mathtt{IWFW}}$ requires independence between any two initial weight parameters. Accordingly, $\Phi_{\mathtt{IWFW}}$ is quantified as follows:
\begin{equation}\label{eq:2}
   \begin{cases}
       \Phi_{\mathtt{IWFW}}(1):\mathsf{EMD} \left ( \mathsf{WD}(W_{0}^{(t)}), \mathsf{GD}(W_{0}^{(t)}) \right) \approx 0 \\
       \Phi_{\mathtt{IWFW}}(2): \mathsf{COV} \left ( W_{0}^{(t)},W_{0}^{(j)} \right) = 0
   \end{cases}
\end{equation}
where $t,j \in [0,T]$, the function $\mathsf{EMD}(\cdot, \cdot)$ represents the Earth Mover’s Distance between two distributions and function $\mathsf{COV}(\cdot, \cdot)$ denotes the covariance of two weight parameters.
We note that $\Phi_{\mathtt{IWFW}}$ comprises two checks: (i) $\Phi_{\mathtt{IWFW}}(1)$ checks whether the initial checkpoint conforms to the prescribed GMM-based random initialization, whereas (ii) $\Phi_{\mathtt{IWFW}}(2)$ evaluates whether the initialization across layers (or parameters) exhibits statistical independence, which is quantified in our implementation via a PCA  statistic. Consequently, $\Phi_{\mathtt{IWFW}}$ is deemed valid only when $\Phi_{\mathtt{IWFW}}(1)\wedge \Phi_{\mathtt{IWFW}}(2)$ holds.

\textbf{$\Phi_{\mathtt{MWCD}}$ (Monotonicity of Weight Checkpoint Distribution).}
Assuming a small learning rate, the variation between consecutive weight checkpoints $W_{i-1}$ and $W_{i}$ should remain bounded. 
We quantify this using the projection distance ($\mathsf{DP}$) between the weight distribution ($\mathsf{WD}$) of the $t$-th layer weights ($W^{(t)}, t \in [0,T]$) as follows:
\begin{equation}
    \mathsf{DP} \left (\mathsf{WD}(W^{(t)}_{i-1}),\mathsf{WD}(W^{(t)}_{i})\right) < \delta , \forall i \in [1,P]
    \label{eq:ts1}
\end{equation}
Specifically, for a linear layer with weights $W^{(t)}$ and bias $\mathsf{b}^{(t)}$, utilizing Euclidean distance instantiates Eq. (\ref{eq:ts1}) as:
\begin{align}
    \sqrt{(\mathsf{b}_{i}^{(t)}- \mathsf{b}_{i-1}^{(t)})^{2}+(||W^{(t)}_{i}||^{2}_{2}-||W^{(t)}_{i-1}||^{2}_{2})^{2}} < \delta
    \label{eq:ts11}
\end{align}
$\Phi_{\mathtt{MWCD}}$ satisfies if the adjacent checkpoint pair holds Eq.~(\ref{eq:ts11}), i.e., any adjacent checkpoint pair is $\delta$-similar.

\subsubsection{Distributed Identity for DNN Model}

Since IR involves the intermediate weight checkpoint of the DNN model, DIDM can be represented by a series of IRs. 
The algorithms of DIDM is shown as follows:

\begin{itemize}
    \item $\mathtt{pp}$ $\gets$ $\mathsf{DIDM.Gen}(1^{\lambda})$: On input the security parameter ($\lambda$), output the public parameters ($\mathtt{pp}$), which include the CRH public parameters ($\mathtt{pp_H}$), commitment public parameters ($\mathtt{pp_C}$) and the zkSNARK public parameters ($\mathtt{pp}_{\mathtt{ZK}}$, $\mathtt{pp}_{\mathtt{in}}$).
    The parameters ($\mathtt{pp}$) are accessible to all functions.

    \item $(\mathtt{IRPK}, \mathtt{IRSK})$ $\gets$ $\mathsf{DIDM.AddrGen}(1^{\lambda})$:
    On input the public parameters ($\mathtt{pp}$), output the model owner's public key address ($\mathtt{IRPK}$) and private key address ($\mathtt{IRSK}$).
    The key address pair $(\mathtt{IRPK}, \mathtt{IRSK})$ is used to bind the IR to the model owner's on-chain identity.

    \item $(\mathtt{pk},\mathtt{vk}) \leftarrow \mathsf{DIDM.KeyGen}(\mathtt{pp})$:
    On input the public parameters ($\mathtt{pp}$), output a proving key ($\mathtt{pk}$) and a verification key ($\mathtt{vk}$).
    We note that the key pair $(\mathtt{pk}, \mathtt{vk})$ involves the key pair $(\mathtt{pk}_{\mathbf{R}}, \mathtt{vk}_{\mathbf{R}})$ for the relation ($\mathbf{R}$) and the key pair $(\mathtt{pk}_{in}, \mathtt{vk}_{in})$ for the internal zkSNARK proof ($\pi_{in}$).

    \item ($[\mathtt{IR}_{i}]_{i \in [0,P]}$, $[\mathtt{CP}_{i}]_{i \in [0,P]}$) $\gets$ $\mathsf{DIDM.IRGen}(\mathtt{pp}, \mathtt{IRPK},\mathbf{W},$ $ (\Phi_{\mathtt{CWCD}},\Phi_{\mathtt{IWFW}},\Phi_{\mathtt{MWCD}}))$: 
    On input the public key address ($\mathtt{IRPK}$), the weight checkpoint sequence ($\mathbf{W}$), and three predicates, it outputs a series of IRs ($[\mathtt{IR}]_{i \in [0,P]}$) and the corresponding commitments ($[\mathtt{CP}_i]_{i \in [0,P]}$).



    \item $\pi_{in}$ $\gets$ $\mathsf{DIDM.InnerZKProve}$($\mathtt{pp},\mathtt{pk}, S_{inp}, J_{in}$): On input the proving key ($\mathtt{pk}$), the witness ($J_{in}$) and the instance ($S_{inp}$), the algorithm outputs an internal zkSNARK proof ($\pi_{in}$).
    The proof ($\pi_{in}$) of internal NIZK relation ($\mathbf{R}_{in}$) asserts the valid opening of $[\mathtt{CP}_i]_{i \in [0,P]}$ and satisfication of predicates over the weight checkpoints ($\mathbf{W}$).


    \item ($\pi_{out}$, $\mathtt{PCP}$) $\gets$ $\mathsf{DIDM.OuterZKProve}$($\mathtt{pp},\mathtt{pk}, \mathtt{vk},S_{inp}, J_{out}, $ $\pi_{in}$): 
    On input the key pair ($\mathtt{pk}, \mathtt{vk}$), the instance ($S_{inp}$), the witness information ($J_{out}$), and the internal proof ($\pi_{in}$), it outputs an external proof ($\pi_{out}$) and a predicate commitment ($\mathtt{PCP}$).
    The proof ($\pi_{out}$) of NIZK relation ($\mathbf{R}$) asserts the predicate satisfiability with data privacy and function privacy.


    \item $b \leftarrow \mathsf{DIDM.Verify}(\mathtt{pp},\mathtt{vk},S_{outp}, \pi_{out})$:
    On input the verification key ($\mathtt{vk}$), the instance ($S_{outp}$), and the external proof ($\pi_{out}$), it outputs $b \in \{0,1\}$ indicating the verification result for $[\mathtt{IR}_i]_{i \in [0,P]}$.


\end{itemize}
Specifically, the first four algorithms ($\mathsf{DIDM.Gen}$, $\mathsf{DIDM.AddrGen}$, $\mathsf{DIDM.KeyGen}$, and $\mathsf{DIDM.IRGen}$) are used for the DIDM generation, whereas the last three algorithms ($\mathsf{DIDM.InnerZKProve}$, $\mathsf{DIDM.OuterZKProve}$, and $\mathsf{DIDM.Verify}$) constitute the DIDM auditing.

\begin{algorithm}[!t]
	\caption{$\mathsf{DIDM.Gen}$ for Parameter Initialization}
	\label{alg:1}
	\begin{algorithmic}[1]
		\REQUIRE{public parameter $\mathtt{pp}$}
  
        \ENSURE{security parameter $\lambda$, NIZK binary relation $\mathbf{R}$}




        \STATE $\mathbf{R}$ is a NIZK relation asserting predicate satisfiability with data and function privacy.

        \STATE Let $\mathbf{R}_{in}$ be the internal NIZK relation in $\mathbf{R}$ asserting: 
         $\forall i \in [0,P]$, $\mathtt{CP}_i$ opens to $(\mathtt{IRPK}$, $W_{i}$,  $\Phi_{\mathtt{CWCD}}$, $\Phi_{\mathtt{IWFW}}$, $\Phi_{\mathtt{MWCD}}$) $\land$ $\Phi_{\mathtt{CWCD}}(\mathbf{W}) = 1 $ $\land $ $\Phi_{\mathtt{IWFW}}(\mathbf{W}) = 1 $ $\land $ $\Phi_{\mathtt{MWCD}}(\mathbf{W}) = 1 $.

        \STATE{$\mathtt{pp_{H}}$ $\gets$ $\mathsf{CRH.Gen}(1^{\lambda})$: Output the public parameters $\mathtt{pp_{H}}$ for collision-resistant hash function.}
        \STATE{$\mathtt{pp_{C}}$ $\gets$ $\mathsf{CS.Gen}(1^{\lambda})$: Output the public parameters $\mathtt{pp_{C}}$ for commitment scheme.}
        


        \STATE{$\mathtt{pp_{ZK}}$ $\gets$ $\mathsf{NIZK.Gen}(1^{\lambda}, \mathbf{R})$: Output the public parameters $\mathtt{pp_{ZK}}$ for relation $\mathbf{R}$.}

        \STATE{$\mathtt{pp_{in}}$ $\gets$ $\mathsf{NIZK.Gen}(1^{\lambda}, \mathbf{R}_{in})$: Output the public parameters $\mathtt{pp_{in}}$ for internal relation $\mathbf{R}_{in}$.}
        
        \STATE{ return $\mathtt{pp}:= (\mathtt{pp_{H}},\mathtt{pp_{C}},\mathtt{pp_{ZK}}, \mathtt{pp_{in}})$.}
    
	\end{algorithmic}

\end{algorithm}

\begin{algorithm}[!t]
	\caption{$\mathsf{DIDM.AddrGen}$ for Address Key Pair Generation of Model Owner}
	\label{alg:2}
	\begin{algorithmic}[1]
		\REQUIRE{address key pair ($\mathtt{IRPK}$, $\mathtt{IRSK}$)}
  
        \ENSURE{public parameter $\mathtt{pp}$  }
        
        \STATE{Sample secret key $sk_{PR}$ for pseudorandom function.}
        \STATE{Sample randomness $cr_{kp}$ for commitment scheme.}

        \STATE{$\mathtt{IRSK}$ $:=(sk_{PR}$, $cr_{kp}$): Construct private key address.}
        
        \STATE{$\mathtt{IRPK}$ $\gets$ $\mathsf{CS.Commit}(\mathtt{pp_C}$, $sk_{PR}$, $cr_{kp}$): Output public key address ($\mathtt{IRPK}$).}

        \STATE{ return address key pair ($\mathtt{IRPK}$, $\mathtt{IRSK}$) of model owner and randomness $cr_{kp}$.}

	\end{algorithmic}

\end{algorithm}

The details of DIDM generation are shown in Alg.~\ref{alg:1} -- Alg.~\ref{alg:4}.
Specifically, Alg.~\ref{alg:1} instantiates the public parameters ($\mathtt{pp}$) for the hash and commitment primitives and generates the corresponding NIZK public parameters for relations ($\mathbf{R}$ and $\mathbf{R}_{in}$).
In particular, A2-DIDM defines the internal NIZK relation ($\mathbf{R}_{in}$) inside $\mathbf{R}$ to explicitly capture the internal consistency constraints (i.e., valid opening of $\mathtt{CP}$ and satisfaction of $\Phi_{\mathtt{CWCD}},\Phi_{\mathtt{IWFW}},\Phi_{\mathtt{MWCD}}$).
Finally, Alg.~\ref{alg:1} outputs $\mathtt{pp} :=(\mathtt{pp_H},\mathtt{pp_C},\mathtt{pp_{ZK}},\mathtt{pp_{in}})$.
Given parameters ($\mathtt{pp}$), Alg.~\ref{alg:2} derives the model owner’s address key pair $(\mathtt{IRPK}, \mathtt{IRSK})$ by committing to the auxiliary information ($sk_{PR}$ and $cr_{kp}$), while Alg.~\ref{alg:3} generates the NIZK proving and verification keys for $\mathbf{R}$ and $\mathbf{R}_{in}$ in a modular fashion. 
Finally, Alg.~\ref{alg:4} shows details for IR generation of a new DNN model.
Specifically, for each checkpoint ($W_i$), it commits to $(\mathtt{IRPK}, W_i, \Phi_{\mathtt{CWCD}}, \Phi_{\mathtt{IWFW}}, \Phi_{\mathtt{MWCD}})$ to obtain IR commitment ($\mathtt{CP}_i$), and assembles $\mathtt{CP}_i$ together with the associated metadata and auxiliary information into a record ($\mathtt{IR}_i$), yielding the record sequence $[\mathtt{IR}_i]_{i\in[0,P]}$ for subsequent DIDM auditing.

\begin{algorithm}[!t]
	\caption{$\mathsf{DIDM.KeyGen}$ for NIZK Key Pair Generation}
	\label{alg:3}
    
	\begin{algorithmic}[1]
		\REQUIRE{Key Pair ($\mathtt{pk}$, $\mathtt{vk}$)}
  
        \ENSURE{public parameter $\mathtt{pp}$}
        
        \STATE{$(\mathtt{pk}_{\mathbf{R}}, \mathtt{vk}_{\mathbf{R}})$ $\gets$ $\mathsf{NIZK.KeyGen}(\mathtt{pp_{ZK}})$: Output the key pair for relation $\mathbf{R}$.}

        \STATE{$(\mathtt{pk}_{in}, \mathtt{vk}_{in})$ $\gets$ $\mathsf{NIZK.KeyGen}(\mathtt{pp_{in}})$: Output the key pair for relation $\mathbf{R}_{in}$.}

        
        \STATE{ return $\mathtt{pk}:= (\mathtt{pk}_{\mathbf{R}},\mathtt{pk}_{in}) $ and $\mathtt{vk}:= (\mathtt{vk}_{\mathbf{R}},\mathtt{vk}_{in})$.}
    
	\end{algorithmic}

\end{algorithm}

\begin{algorithm}[!t]
	\caption{$\mathsf{DIDM.IRGen}$ for IR Generation}
	\label{alg:4}
    
	\begin{algorithmic}[1]
		\REQUIRE{a series of new DNN model identity records $[\mathtt{IR}_{i}]_{i \in [0,P]}$ for DIDM}
  
        \ENSURE{public parameter $\mathtt{pp}$; 
        public key address $\mathtt{IRPK}$; \\
        predicates $\Phi_{\mathtt{CWCD}}$,  $\Phi_{\mathtt{IWFW}}$, $\Phi_{\mathtt{MWCD}}$; randomness $cr_{kp}$;\\
        weight checkpoints $\mathbf{W}=[W_i]_{i \in [0,P]}$}

        \STATE{$\mathtt{CP}_{i}$ $\gets$  $\mathsf{CS.Commit}(\mathtt{pp_{C}}, \mathtt{IRPK}_{i}|| W_{i}||\Phi_{\mathtt{CWCD}}||\Phi_{\mathtt{IWFW}}||\Phi_{\mathtt{MWCD}};$ $ cr_{kp})$: Output the commitment $\mathtt{CP}_{i}$ for $\mathtt{IR}_{i}$.}
        \STATE{$\mathtt{IR}$$_{i}$ = $ ( \mathtt{CP}_{i}, \mathtt{IRPK}, W_{i}, (\Phi_{\mathtt{CWCD}},\Phi_{\mathtt{IWFW}},\Phi_{\mathtt{MWCD}}), \mathtt{AUX})$: Construct the identity record $\mathtt{IR}$$_{i}$ for $i$-th weight checkpoints $W_{i}$, where $\mathtt{AUX}$ includes the randomness $cr_{kp}$. }
        \STATE{[$\mathtt{IR}$$_{i}]_{i \in [0,P]}$ = $[[\mathtt{CP}_{i}],$ $\mathtt{IRPK},$ $[W_{i}],$$ (\Phi_{\mathtt{CWCD}},$ $\Phi_{\mathtt{IWFW}},$ $\Phi_{\mathtt{MWCD}}); $ $\mathtt{AUX}]_{i \in [0,P]}$: Construct the $[\mathtt{IR}_{i}]_{i \in [0,P]}$ of $\mathbf{W}$.}
        \STATE{Return [$\mathtt{IR}$$_{i}]_{i \in [0,P]}$ for the DIDM.}
    
	\end{algorithmic}

\end{algorithm}

\subsection{Accumulator-based Auditing for DIDM}

In A2-DIDM, the blockchain traces and tracks the IR of DNN model. 
A complete DIDM is represented by a collection of IRs, and three predicates control the uniqueness of the DNN weight checkpoint sequence. $\mathcal{P}$ performs DNN training off-chain, and then submits transactions to blockchain for verification, which include the commitment of training history (i.e., weight checkpoint sequences). 
The execution correctness is ensured by efficiently opening new IR commitments, which allows for the validation of the performed computations. 
Additionally, the computational integrity is captured in the satisfiability of three predicates for all new IRs.



\begin{figure}[!t]
    \centering
    \includegraphics[width=1.0\columnwidth]{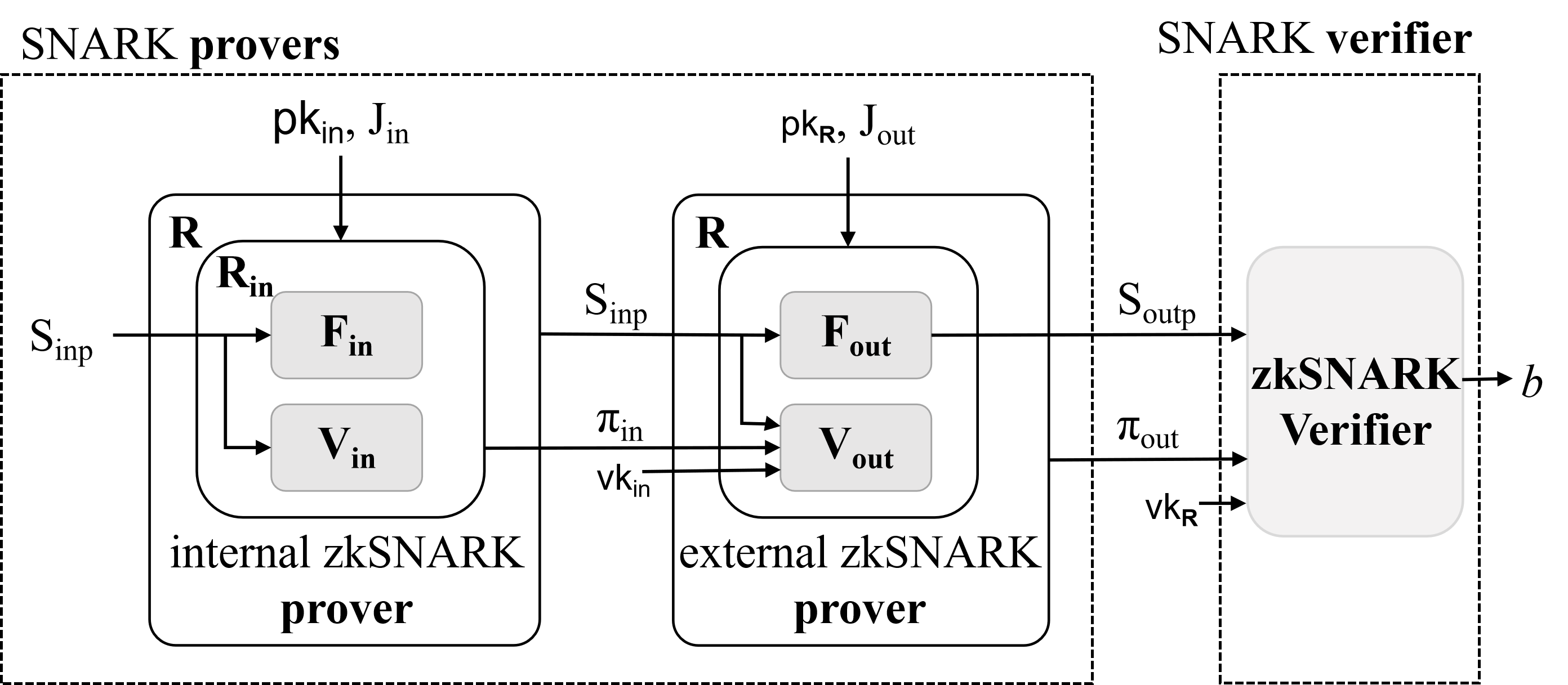}
    \caption{
    The circuit of predicate satisfiability (captured by relation $\mathbf{R}$) without accumulator. 
    $\mathbf{F_{in}}$ and $\mathbf{V_{in}}$ denote the proving and verification functions in Alg.~\ref{alg:5} (lines 5 and 4), respectively. Similarly, $\mathbf{F_{out}}$ and $\mathbf{V_{in}}$ denote the commit and verification function in Alg.~\ref{alg:6} (lines 6 and 5).
    }
    \label{fig3}
\end{figure}

To preserve the zero-knowledge of model weights and predicates, the predicate satisfiability of DIDM audit in A2-DIDM is instantiated as a two-step zkSNARKs shown in Fig.~\ref{fig3}. 
The internal zkSNARK \textbf{prover} first derives the public instance ($S_{inp}$) from the commitment components of the IRs and auxiliary information, and assembles the witness ($J_{in}$) from the remaining secret elements. 
It then generates an internal proof ($\pi_{in}$) for relation $\mathbf{R}_{in}$, which attests that the commitments open correctly and that the three predicates are satisfied. 
The external zkSNARK \textbf{prover} subsequently binds the internal verification key ($\mathtt{vk}_{in}$) by computing a collision-resistant hash and committing to the predicate commitment.
Moreover, it generates an external proof ($\pi_{out}$) for relation $\mathbf{R}$ asserting that internal proof ($\pi_{in}$) verifies on the instance ($S_{inp}$) under the bound key. 
The zkSNARK \textbf{verifier} finally checks only proof ($\pi_{out}$) on the compact public instance ($S_{outp}$) and outputs an accept/reject bit ($b$), while the internal witness remains hidden.


\begin{algorithm}[!t]
	\caption{$\mathsf{DIDM.InnerZKProve}$ for Internal zkSNARK \textbf{Prover} }
	\label{alg:5}
    
	\begin{algorithmic}[1]
		\REQUIRE{internal proof $\pi_{in}$ for $\mathbf{R}_{in}$}
  
        \ENSURE{identity records $[\mathtt{IR}_{i}]_{i \in [0,P]}$; 
        public parameters $\mathtt{pp}$
        }

        \STATE{Parse $[\mathtt{IR}_{i}]_{i \in [0,P]}$ $:=$  $([\mathtt{CP}$$_{i}]_{i \in [0,P]}$, $\mathtt{IRPK}, \mathbf{W},\Phi_{\mathtt{CWCD}},$ $ \Phi_{\mathtt{IWFW}},$ $ \Phi_{\mathtt{MWCD}}; $ $\mathtt{AUX}$).}

        \STATE{$S_{inp}$ = ($[\mathtt{CP}$$_{i}]_{i \in [0,P]}; \mathtt{AUX}$): Construct the instance $S_{inp}$.}
        
        \STATE{$J_{in}$ = $(\mathtt{IRPK}, \mathbf{W},\Phi_{\mathtt{CWCD}}, \Phi_{\mathtt{IWFW}}, \Phi_{\mathtt{MWCD}})$: Construct the witness $J_{in}$.}

         \STATE{For $\mathbf{R}_{in}$, accept iff $\mathsf{CS.Open}(\mathtt{pp_C},S_{inp},J_{in})=1$ $\land$ $\Phi_{\mathtt{CWCD}}(\mathbf{W}) = 1 $ $\land $ $\Phi_{\mathtt{IWFW}}(\mathbf{W}) = 1 $ $\land $ $\Phi_{\mathtt{MWCD}}(\mathbf{W}) = 1 $.}

        \STATE{$\pi_{in}$ $\gets$ $\mathsf{NIZK.Prove}$($\mathtt{pp}_{\mathtt{in}}, S_{inp}, J_{in}; \mathtt{pk}_{in}$): Output the internal zkSNARK proof.}

        \STATE{return internal proof $\pi_{in}$ and instance $S_{inp}$.}
     
	\end{algorithmic}

\end{algorithm}

 The concrete procedures corresponding to Fig.~\ref{fig3} are formalized in Alg.~\ref{alg:5} -- Alg.~\ref{alg:7}. Specifically, Alg.~\ref{alg:5} ($\mathsf{DIDM.InnerZKProve}$) parses the IRs to construct $(S_{inp},J_{in})$ and outputs $\pi_{in}$ after enforcing commitment opening and predicate satisfaction with respect to $\mathbf{R}_{in}$; if any predicate fails, no valid internal proof exists and the audit cannot proceed. 
To reduce the cost of handling large keys, Alg.~\ref{alg:6} ($\mathsf{DIDM.OuterZKProve}$) commits to the collision-resistant hash of $\mathtt{vk}_{in}$ and generates predicate commitment ($\mathtt{PCP}$).
It outputs $\pi_{out}$ that certifies the verification of $\pi_{in}$, thereby binding the IR commitment to the predicate commitment under $\mathbf{R}$. 
Alg.~\ref{alg:7} ($\mathsf{DIDM.Verify}$) forms $S_{outp}$ from $\mathtt{PCP}$ and $\mathtt{CP}$ and verifies $\pi_{out}$ using $\mathtt{vk}_{\mathbf{R}}$, returning a single bit that represents the DIDM audit result.


\begin{algorithm}[!t]
	\caption{$\mathsf{DIDM.OuterZKProve}$ for External zkSNARK \textbf{Prover}}
	\label{alg:6}
	\begin{algorithmic}[1]
		\REQUIRE{external proof $\pi_{out}$ for $\mathbf{R}$}
  
        \ENSURE{instance $S_{inp}$; internal proof $\pi_{in}$; key pair ($\mathtt{pk}$, $\mathtt{vk}$); public parameters $\mathtt{pp}$}

        \STATE Parse $\mathtt{pk}:=(\mathtt{pk}_{\mathbf{R}}, \mathtt{pk}_{in})$, $\mathtt{vk}:=(\mathtt{vk}_{\mathbf{R}}, \mathtt{vk}_{in})$.

        \STATE Parse $S_{inp}$ $:=$ ($[\mathtt{CP}$$_{i}]_{i \in [0,P]}; \mathtt{AUX}$), where $\mathtt{AUX}$ includes randomness $cr_{kp}$.

        \STATE{$\mathtt{h}(\mathtt{vk}_{in})$ $\gets$ $\mathsf{CRH.Val}(\mathtt{pp_H}, \mathtt{vk}_{in})$: Output the collision-resistant hash of verification key.}

        \STATE{$J_{out}$ = $ \mathtt{h}(\mathtt{vk}_{in})$: Construct the witness $J_{out}$.}

        \STATE{For $\mathbf{R}$, accept iff $\mathsf{NIZK.Verify}(\mathtt{pp_{in}},S_{inp},\pi_{in}; \mathtt{vk}_{in})=1$.}

        \STATE{$\mathtt{PCP}$ $\gets$ $\mathsf{CS.Commit}$($\mathtt{pp_C}, \mathtt{h}(\mathtt{vk}_{in}); cr_{kp}$): Output predicate commitment.}

        \STATE{$\pi_{out}$ $\gets$ $\mathsf{NIZK.Prove}$($\mathtt{pp}_{\mathtt{zk}}, S_{inp}, J_{in}; \mathtt{pk}_{\mathbf{R}}$): Output the external zkSNARK proof.}

        \STATE{return external proof $\pi_{out}$ and predicate commitment $\mathtt{PCP}$.}

	\end{algorithmic}

\end{algorithm}

\begin{algorithm}[!t]
	\caption{$\mathsf{DIDM.Verify}$ for zkSNARK \textbf{Verifier}}
	\label{alg:7}
	\begin{algorithmic}[1]
		\REQUIRE{verify result for DIDM audit}
  
        \ENSURE{verification key $\mathtt{vk}_{\mathbf{R}}$; public parameters $\mathtt{pp}$; external proof $\pi_{out}$; predicate commitment $\mathtt{PCP}$; IR commitment $[\mathtt{CP}_{i}]_{i \in [0,P]}$; aux}

        \STATE{$S_{outp}$ = $\left [ \mathtt{PCP}, [\mathtt{CP}_{i}]_{i \in [0,P]}\right ]$: Construct an instance $S_{outp}$.}

        \STATE{$b$ $\gets$ $\mathsf{NIZK.Verifiy}(\mathtt{pp_{ZK}},S_{outp},\pi_{out}; \mathtt{vk}_{\mathbf{R}})$: Verify the external proof $\pi_{out}$ and output a bit $b \in \{0,1\}$, where $b=1$ denotes DIDM audit accepts and $b=0$ denotes rejects.}

        \STATE{return the DIDM audit results.}
     
	\end{algorithmic}

\end{algorithm}


In Fig.~\ref{fig3}, the external zkSNARK \textbf{prover} should implement an in-circuit zkSNARK \textbf{verifier} to execute $\mathbf{V_{out}}$, where polynomial commitment verification typically involves expensive pairing operations. 
To reduce these pairing costs, A2-DIDM incorporates an accumulator so that the dominant cost of $\mathbf{V_{out}}$ is deferred to the DIDM \textbf{verifier}, allowing the DIDM \textbf{prover} to avoid a heavy verification circuit.
Fig.~\ref{fig4} illustrates the accumulator-based DIDM \textbf{prover} and \textbf{verifier} construction.
The DIDM \textbf{prover} embeds only the accumulator logic inside the DIDM audit circuit, which is substantially lighter than embedding a full zkSNARK \textbf{verifier} for $\mathbf{V_{out}}$. 
Concretely, the external zkSNARK \textbf{prover} takes the internal proof ($\pi_{in}$) and an empty accumulator ($\mathtt{ACCU}_{in} = \bot$) as input, and partially verifies the proof ($\pi_{in}$) by checking all components other than the pairing operations.
The accumulator \textbf{prover} derives the group elements required for efficient pairing operations, while the accumulator \textbf{verifier} validates the correctness of this accumulated computation. 
The DIDM \textbf{verifier} then invokes an additional \textbf{decider} procedure to complete the remaining pairing operations. 
Overall, the accumulator shifts a small number of costly pairings from the external zkSNARK \textbf{prover} side to on-chain execution by the \textbf{decider}.


\begin{figure}[!t]
    \centering
    \includegraphics[width=1.0\columnwidth]{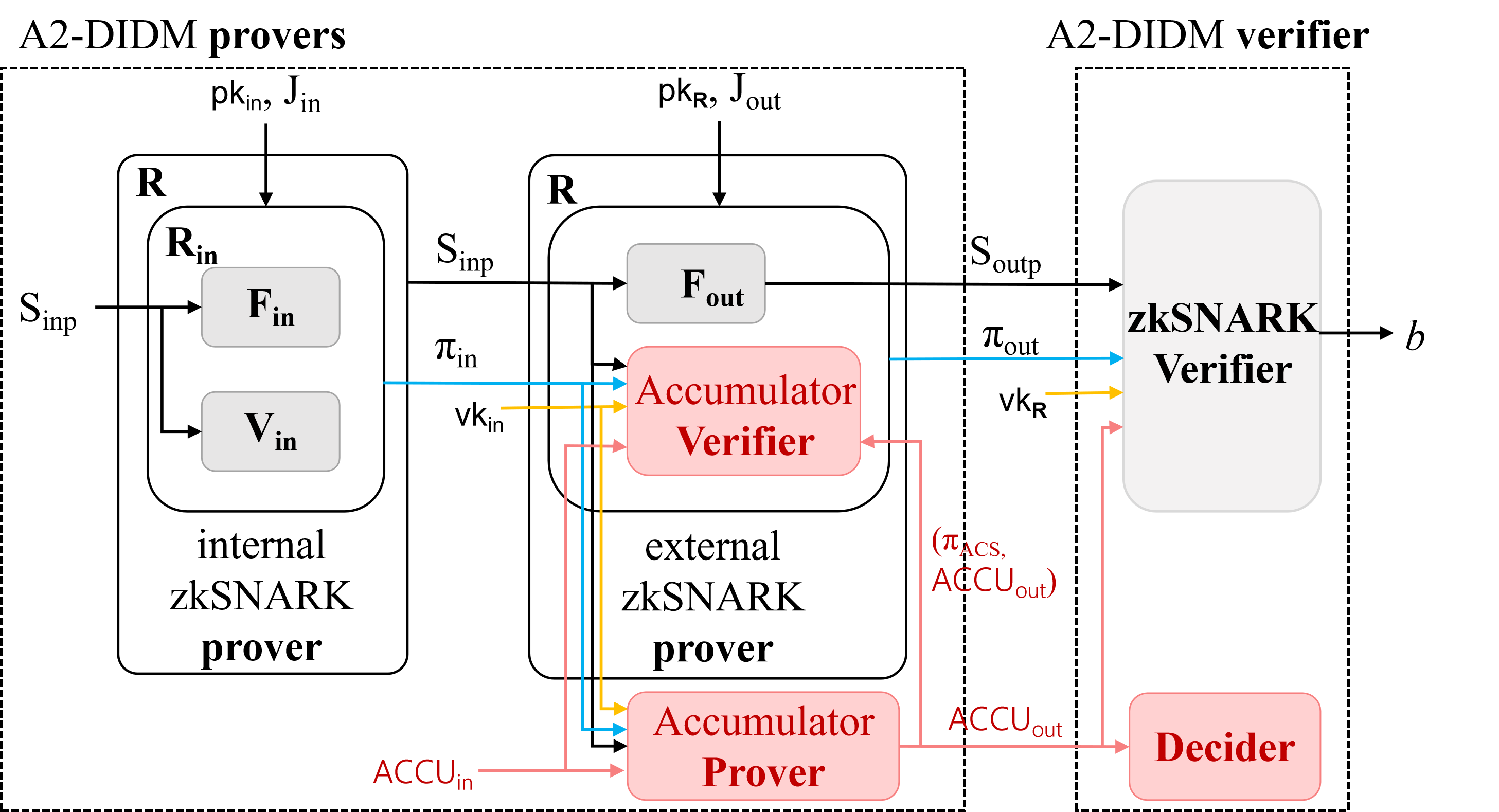}
    \caption{The construction of accumulator-based DIDM \textbf{provers} and \textbf{verifiers}.
    Red boxes are accumulator \textbf{prover}, \textbf{verifier} and \textbf{decider}, which are differet from Fig. \ref{fig3}. The yellow and blue lines are to distinguish the input of \textbf{prover}/\textbf{verifier}, and have no practical meaning. 
    Accumulator keys ($\mathtt{apk}$, $\mathtt{avk}$, $\mathtt{dk}$) are dropped from the figure for visual clarity. 
    }
    \label{fig4}
\end{figure}

To be specific, the accumulator scheme ($\mathsf{ACS}$) involves a tuple of PPT algorithms $\mathsf{ACS}$ $=$ $(\mathsf{ACS.Gen},$ $\mathsf{ACS.KeyGen},$ $\mathsf{ACS.Prove},$ $\mathsf{ACS.Verify},$ $\mathsf{ACS.Decide})$ as follows:
\begin{itemize}
    \item $\mathtt{pp_{ACS}}$ $\gets$ $\mathsf{ACS.Gen}(1^{\lambda})$: On input the security parameter ($\lambda$), output the public parameter ($\mathtt{pp_{ACS}}$) for accumulator scheme.
    
    \item ($\mathtt{apk}$, $\mathtt{avk}$, $\mathtt{dk}$) $\gets$ $\mathsf{ACS.KeyGen}(\mathtt{pp_{ACS}})$: On input public parameters ($\mathtt{pp_{ACS}}$), output an accumulator \textbf{prover} key ($\mathtt{apk}$), an accumulator \textbf{verifier} key ($\mathtt{avk}$) and \textbf{decider} key ($\mathtt{dk}$). 
    
    \item ($\mathtt{ACCU}_{out},\pi_{ACS}$) $\gets$ $\mathsf{ACS.Prove}( \mathtt{pp_{ACS}},\mathtt{apk}, Q_{in},\mathtt{vk}_{in},$ $ (\mathtt{ACCU}_{in}, \pi_{in}) )$: 
    This algorithm is executed by accumulator \textbf{prover}. It inputs accumulator \textbf{prover} key ($\mathtt{apk}$), verification key ($\mathtt{vk}_{in}$), an old accumulator proof tuple $(\mathtt{ACCU}_{in},  \pi_{in})$ and a new accumulator instance ($Q_{in}$).  
    It outputs a new accumulator ($\mathtt{ACCU}_{out}$) and accumulator proof ($\pi_{ACS}$) for $Q_{in}$. 
    
    \item $b \gets \mathsf{ACS.Verify}(\mathtt{pp_{ACS}},\mathtt{avk},\mathtt{vk}_{in}, Q_{in}, (\mathtt{ACCU}_{in}, \pi_{in}), $ $(\mathtt{ACCU}_{out},\pi_{ACS}))$:
    This algorithm is executed by accumulator \textbf{verifier}. 
    It checks the correctness of the accumulation step without executing the final expensive pairing operations.

    \item $b'$ $\gets$ $\mathsf{ACS.Decide} \left( \mathtt{pp_{ACS}}, \mathtt{dk}, \mathtt{ACCU}_{out} \right)$: 
    This algorithm is executed by \textbf{decider}. 
    It completes the deferred heavy verification, such as pairing operations for a polynomial commitment opening.
    
\end{itemize}
In Alg.~\ref{alg:6} (line~5), we replace the original zkSNARK \textbf{verifier} in $\mathbf{V_{out}}$ with an accumulator \textbf{verifier}; namely, $\mathbf{V_{out}}$ accepts iff
$\mathsf{ACS.Verify}(\mathtt{pp_{ACS}},$ $\mathtt{avk},$ $\mathtt{vk}_{in},$ $ Q_{in},(\mathtt{ACCU}_{in}, \pi_{in}),$ $(\mathtt{ACCU}_{out},\pi_{ACS}))=1$.
Moreover, DIDM \textbf{prover} initializes the accumulator to the empty value for the first accumulation, i.e., $\mathtt{ACCU}_{in}\gets\bot$. 
The accumulator instance $Q_{in}$ consists of polynomial-commitment-related information, including the polynomial commitment, the degree bound, the evaluation point, the claimed value, and the opening proof~\cite{bunz2020proof}. 
Finally, the DIDM \textbf{verifier} accepts iff both the zkSNARK \textbf{verifier} and the \textbf{decider} output the decision bit~$1$.

To further reduce the accumulator \textbf{verifier}’s cost, we leverage the linear homomorphism and batch-opening properties of pairing-based polynomial commitments (e.g., KZG) to batch the three predicate proofs ($\pi_{\Phi}$) and the commitment in $\mathbf{R}_{in}$ (Alg.~\ref{alg:5}, line~4).
This process ultimately outputs the batched internal proof ($\pi_{in}$).
Specifically, when verifying $P$ proof instances, each instance consists of a polynomial commitment (i.e., IR commitment $\mathtt{CP}_i$), an evaluation point (i.e., weight checkpoint $W_i$), an evaluation value ($v_i$), and a predicate proof ($\pi_{\Phi,i}$) that attests to the relation $\mathsf{Open}(\mathtt{CP}_i, W_i)=v_i$. 
The accumulator \textbf{verifier} first samples a random challenge scalar $\rho \in \mathbb{F}$, which is shared across all instances.
Leveraging the linear homomorphism of the commitment scheme, the \textbf{verifier} does not check the pairing equations individually but instead constructs a random linear combination. 
For instances sharing the same evaluation point ($W_i$), the \textbf{verifier} computes the batched commitment $\mathtt{CP}_{B} = \sum_{i=0}^P \rho^{i} \cdot \mathtt{CP}_i$ and the batched proof $\pi_{in} = \sum_{i=0}^P \rho^{i} \cdot \pi_{\Phi,i}$. 
According to the Schwartz-Zippel Lemma, if the batched equation $\mathsf{Open}(\mathtt{CP}_{B}, \mathbf{W}) = \sum \rho^{i} \cdot v_i$ holds, it guarantees with overwhelming probability that all predicate proofs are valid.

Moreover, the accumulator \textbf{verifier} only executes lightweight group operations, such as scalar multiplications and group additions, to compute the above linear combination, whereas the most expensive bilinear pairing checks are deferred to the on-chain \textbf{decider}. 
Consequently, the external zkSNARK \textbf{prover} does not need to encode pairing computations as circuit constraints.
Specifically, when verifying a zkSNARK built on a pairing-based polynomial commitment scheme, the \textbf{verifier} derives two checkpoints in $2\mathbb{G}_{1}$ that are used for the final bilinear pairing check. 
The accumulator in A2-DIDM exposes these $2\mathbb{G}_{1}$ checkpoints as public inputs. 
In particular, the accumulator value ($\mathtt{ACCU}_{out}$) involves the $2\mathbb{G}_{1}$ pairing checkpoints that summarize the aggregated validity of all batched proofs. 
The on-chain \textbf{decider} then verifies these checkpoints via a single pairing check.
Accordingly, the verification of $P$ predicate proofs and commitments reduces to $O(P)$ off-chain group operations incurs only $O(P)$ off-chain group operations, while the on-chain \textbf{decider} verifies the accumulated result with a single pairing check, which suffices to validate the batched proof.

For each transaction, the on-chain DIDM \textbf{verifier} performs a constant number of checks on the external zkSNARK proof and executes a constant number of checks via the \textbf{decider}, jointly completing the final pairing verification. 
Both the external proof and the accumulator have constant sizes, and the number of expensive group operations executed by the zkSNARK \textbf{verifier} and \textbf{decider} does not increase with the volume of off-chain computations. 
Therefore, the overall on-chain verification overhead is independent of the transaction content and remains constant. 
A2-DIDM achieves succinct proofs and a lightweight \textbf{verifier} circuit through the accumulator.




\section{Theoretical Analysis}\label{sec:ana}

\subsection{Security of Accumulator-based DIDM Audit}

\textbf{Completeness of accumulator-based audit scheme.} 
We define the completeness of the DIDM audit scheme via the experiment ($\mathsf{Exp}^{\mathcal{C}}_{\mathsf{DIDM}}$). 
For all valid $\mathbf{W}$, the accumulator-based audit scheme is complete if the following Eq.~\ref{eq:comp} holds:

\begin{equation}
\Pr\big[\mathsf{Exp}^{\mathcal{C}}_{\mathsf{DIDM}}(1^\lambda)=1\big]
\ge 1-\mathsf{negl}(\lambda).
\label{eq:comp}
\end{equation}

\begin{mdframed}
\textbf{$\mathsf{Exp}^{\mathcal{C}}_{\mathsf{DIDM}}(1^\lambda)$:}
\begin{enumerate}
\item Sample parameters and compute key pairs:
\[
\left\{
\begin{aligned}
&\mathtt{pp} \gets \mathsf{DIDM.Gen}(1^\lambda), \mathtt{pp_{ACS}} \gets \mathsf{ACS.Gen}(1^\lambda)\\
&(\mathtt{IRPK},\mathtt{IRSK}) \gets \mathsf{DIDM.AddrGen}(\mathtt{pp})\\
&(\mathtt{pk},\mathtt{vk})\gets \mathsf{DIDM.KeyGen}(\mathtt{pp})\\
&(\mathtt{apk},\mathtt{avk},\mathtt{dk}) \gets \mathsf{ACS.KeyGen}(\mathtt{pp_{ACS}})
\end{aligned}
\right.
\]

\item Compute IR and commitment:
\begin{align*}
    \left(
    \begin{matrix}
        [\mathtt{IR}_i]_{i\in[0,P]}\\
        [\mathtt{CP}_i]_{i\in[0,P]}
    \end{matrix}
    \right)
\gets\mathsf{DIDM.IRGen}\left(
\begin{matrix}
    \mathtt{pp}, \mathtt{IRPK}, \Phi_{\mathtt{CWCD}},\\ 
    \mathbf{W},\Phi_{\mathtt{IWFW}}, \Phi_{\mathtt{MWCD}}
\end{matrix}
\right)
\end{align*}

\item Compute internal zkSNARK proof:
\begin{align*}
    \pi_{in} \gets \mathsf{DIDM.InnerZKProve}\left( 
    \begin{matrix}
        \mathtt{pp},\mathtt{pk},
        S_{inp},J_{in}
    \end{matrix}
    \right)
\end{align*}

\item Compute accumulator proof:
$$
\left(
\begin{matrix}
    \mathtt{ACCU}_{out},\\
    \pi_{ACS} 
\end{matrix}
\right) \gets 
\mathsf{ACS.Prove}\left(
\begin{matrix}
    \mathtt{pp_{ACS}},\mathtt{apk},Q_{in},\\\mathtt{vk}_{in},
    \mathtt{ACCU}_{in},\pi_{in}
\end{matrix}
\right)
$$

\item Compute external zkSNARK proof: 
$$
\left(
\begin{matrix}
    \pi_{out},\\
    \mathtt{PCP}
\end{matrix}
\right) \gets 
\mathsf{DIDM.OuterZKProve}
\left(\begin{matrix}
\mathtt{pp}, \mathtt{pk},\pi_{in},\\
\mathtt{vk}, S_{inp},J_{out}
\end{matrix}\right)
$$

\item Output $1$ iff all checks accept:
\[
\left\{
\begin{aligned}
&\Phi_{\mathtt{CWCD}}(\mathbf{W})= 
\Phi_{\mathtt{IWFW}}(\mathbf{W})= 
\Phi_{\mathtt{MWCD}}(\mathbf{W})=1\\
&\mathsf{CS.Open}(\mathtt{pp},S_{inp},J_{in})=1\\
&\mathsf{DIDM.Verify}(\mathtt{pp},\mathtt{vk},S_{outp},\pi_{out})=1\\
&\mathsf{ACS.Verify}\left(
\begin{matrix} 
\mathtt{pp_{ACS}},\mathtt{avk},\mathtt{vk},
\mathtt{ACCU}_{in},\\\pi_{in},Q_{in},
\mathtt{ACCU}_{out},\pi_{ACS}
\end{matrix}\right ) =1\\
&\mathsf{ACS.Decide}(\mathtt{pp_{ACS}},\mathtt{dk},\mathtt{ACCU}_{out})=1
\end{aligned}
\right.
\]

Otherwise output $0$.
\end{enumerate}
\end{mdframed}

\textbf{Knowledge soundness of accumulator-based audit scheme.}
We define the knowledge soundness of DIDM via the experiment
($\mathsf{Exp}^{\mathcal{KS}}_{\mathsf{DIDM},\mathcal{A}}$). 
For all PPT adversaries $\mathcal{A}$, DIDM is knowledge-sound
if Eq.~\ref{eq:ks} holds:

\begin{equation}
\Pr\big[\mathsf{Exp}^{\mathcal{KS}}_{\mathsf{DIDM},\mathcal{A}}(1^\lambda)=1\big]
\le \mathsf{negl}(\lambda).
\label{eq:ks}
\end{equation}

\begin{mdframed}
\textbf{$\mathsf{Exp}^{\mathcal{KS}}_{\mathsf{DIDM},\mathcal{A}}(1^\lambda)$:}
\begin{enumerate}
\item Sample public parameters and verification keys:
\[
\left\{
\begin{aligned}
&\mathtt{pp} \gets \mathsf{DIDM.Gen}(1^\lambda),
\mathtt{pp_{ACS}} \gets \mathsf{ACS.Gen}(1^\lambda)\\
&(\mathtt{IRPK},\mathtt{IRSK}) \gets \mathsf{DIDM.AddrGen}(\mathtt{pp})\\
&(\mathtt{pk},\mathtt{vk})\gets \mathsf{DIDM.KeyGen}(\mathtt{pp})\\
&(\mathtt{apk},\mathtt{avk},\mathtt{dk}) \gets \mathsf{ACS.KeyGen}(\mathtt{pp_{ACS}})
\end{aligned}
\right.
\]

\item The adversary outputs a purported proof transcript:
\[
(\pi_{out}^\star,\pi_{ACS}^\star,\mathtt{ACCU}_{out}^\star,S_{outp}^\star)
\gets \mathcal{A}\left(
\begin{matrix}
    \mathtt{pp},\mathtt{pp_{ACS}},\\
    \mathtt{vk},\mathtt{avk},\mathtt{dk}
\end{matrix}
\right)
\]

\item The adversary runs a PPT extractor $\mathcal{E}$ to output witnesses and instances:
\[
\left(
\begin{matrix}
    S_{inp}^\star,\mathtt{ACCU}_{in}^\star,\\
     J_{out}^\star,J_{in}^\star,\pi_{in}^\star
\end{matrix}
\right)
\gets \mathcal{E}^{\mathcal{A}}\left(
\begin{matrix}
    \mathtt{pp},\mathtt{pp_{ACS}},
    \mathtt{vk},\pi_{out}^\star,S_{outp}^\star\\
    \mathtt{dk},\mathtt{avk},\pi_{ACS}^\star,
    \mathtt{ACCU}_{out}^\star 
\end{matrix}
\right)
\]

\item Accept iff all public verifications pass:
\[
\left\{
\begin{aligned}
&\mathsf{DIDM.Verify}(\mathtt{pp},\mathtt{vk},S_{outp}^\star,\pi_{out}^\star)=1\\
&\mathsf{ACS.Verify}\left(
\begin{matrix}
    \mathtt{pp_{ACS}},\mathtt{avk},\mathtt{vk},Q_{in},\pi_{in}^\star, \\
    \mathtt{ACCU}_{in}^\star,\mathtt{ACCU}_{out}^\star,\pi_{ACS}^\star 
\end{matrix}\right)=1\\
&\mathsf{ACS.Decide}(\mathtt{pp_{ACS}},\mathtt{dk},\mathtt{ACCU}_{out}^\star)=1
\end{aligned}
\right.
\]

\item Output $1$ if the experiment accepts but no valid witness exists, i.e., \\
$
(S_{inp}^\star,J_{in}^\star) \notin \mathbf{R_{in}}$ $\wedge$ $
((S_{inp}^\star,\pi_{in}^\star),J_{out}^\star)\notin \mathbf{R}
$

Otherwise output $0$.
\end{enumerate}
\end{mdframed}

\textbf{Zero-knowledge of accumulator-based DIDM audit scheme.}
We define the zero-knowledge of DIDM via the experiment
($\mathsf{Exp}^{\mathcal{ZK}}_{\mathsf{DIDM},\mathcal{A}}$).
For all PPT adversaries ($\mathcal{A}$), DIDM is zero-knowledge if there exists a PPT simulator ($\mathcal{S}$) such that Eq.~\ref{eq:zk} holds:

\begin{equation}
\begin{aligned}
| \Pr[\mathsf{Exp}^{\mathcal{ZK}\text{-}\mathsf{Real}}_{\mathsf{DIDM},\mathcal{A}}(1^\lambda)=1]
 - \Pr[\mathsf{Exp}^{\mathcal{ZK}\text{-}\mathsf{Sim}}_{\mathsf{DIDM},\mathcal{A},\mathcal{S}}&(1^\lambda)=1] | \\
&\le \mathsf{negl}(\lambda).
\end{aligned}
\label{eq:zk}
\end{equation}

\begin{mdframed}
\textbf{$\mathsf{Exp}^{\mathcal{ZK}\text{-}\mathsf{Real}}_{\mathsf{DIDM},\mathcal{A}}(1^\lambda)$:}

\begin{enumerate}
\item Sample parameters and compute key pairs:

\[
\left\{
\begin{aligned}
&\mathtt{pp} \gets \mathsf{DIDM.Gen}(1^\lambda), \mathtt{pp_{ACS}} \gets \mathsf{ACS.Gen}(1^\lambda)\\
&(\mathtt{IRPK},\mathtt{IRSK}) \gets \mathsf{DIDM.AddrGen}(\mathtt{pp})\\
&(\mathtt{pk},\mathtt{vk})\gets \mathsf{DIDM.KeyGen}(\mathtt{pp})\\
&(\mathtt{apk},\mathtt{avk},\mathtt{dk}) \gets \mathsf{ACS.KeyGen}(\mathtt{pp_{ACS}})
\end{aligned}
\right.
\]

\item Compute IR and commitment:
\begin{align*}
    \left(
    \begin{matrix}
        [\mathtt{IR}_i]_{i\in[0,P]}\\
        [\mathtt{CP}_i]_{i\in[0,P]}
    \end{matrix}
    \right)
\gets\mathsf{DIDM.IRGen}\left(
\begin{matrix}
    \mathtt{pp}, \mathtt{IRPK}, \Phi_{\mathtt{CWCD}},\\ 
    \mathbf{W},\Phi_{\mathtt{IWFW}}, \Phi_{\mathtt{MWCD}}
\end{matrix}
\right)
\end{align*}

\item Compute internal zkSNARK proof:
\begin{align*}
    \pi_{in} \gets \mathsf{DIDM.InnerZKProve}\left( 
    \begin{matrix}
        \mathtt{pp},\mathtt{pk},
        S_{inp},J_{in}
    \end{matrix}
    \right)
\end{align*}

\item Compute accumulator proof:
$$
\left(
\begin{matrix}
    \mathtt{ACCU}_{out},\\
    \pi_{ACS} 
\end{matrix}
\right) \gets 
\mathsf{ACS.Prove}\left(
\begin{matrix}
    \mathtt{pp_{ACS}},\mathtt{apk},Q_{in},\\\mathtt{vk}_{in},
    \mathtt{ACCU}_{in},\pi_{in}
\end{matrix}
\right)
$$

\item Compute external zkSNARK proof:
$$
\left(
\begin{matrix}
    \pi_{out},\\
    \mathtt{PCP}
\end{matrix}
\right) \gets 
\mathsf{DIDM.OuterZKProve}
\left(\begin{matrix}
\mathtt{pp}, \mathtt{pk},\pi_{in},\\
\mathtt{vk}, S_{inp},J_{out}
\end{matrix}\right)
$$

\item $\mathcal{A}$ receives the real transcript and outputs a bit:
\[
b \gets \mathcal{A}\left(
\begin{matrix}
\mathtt{pp},\mathtt{pp_{ACS}},
\mathtt{vk},\mathtt{avk},[\mathtt{IR}_i,\mathtt{CP}_i]_{i\in[0,P]},\\
\mathtt{dk},\pi_{in},\pi_{ACS},\mathtt{ACCU}_{out},
\pi_{out},S_{outp}
\end{matrix}
\right)
\]

\item Output $1$ iff $b=1$.
\end{enumerate}
\end{mdframed}

\begin{mdframed}
\textbf{$\mathsf{Exp}^{\mathcal{ZK}\text{-}\mathsf{Sim}}_{\mathsf{DIDM},\mathcal{A},\mathcal{S}}(1^\lambda)$:}
\begin{enumerate}
\item Sample parameters and verification keys:
\[
\left\{
\begin{aligned}
&\mathtt{pp} \gets \mathsf{DIDM.Gen}(1^\lambda), \mathtt{pp_{ACS}} \gets \mathsf{ACS.Gen}(1^\lambda)\\
&(\mathtt{IRPK},\mathtt{IRSK}) \gets \mathsf{DIDM.AddrGen}(\mathtt{pp})\\
&(\mathtt{pk},\mathtt{vk})\gets \mathsf{DIDM.KeyGen}(\mathtt{pp})\\
&(\mathtt{apk},\mathtt{avk},\mathtt{dk}) \gets \mathsf{ACS.KeyGen}(\mathtt{pp_{ACS}})
\end{aligned}
\right.
\]

\item The simulator $\mathcal{S}$ outputs a simulated transcript without using any witness:
\[
\left(
\begin{matrix}
[\mathtt{IR}_i^\star,\mathtt{CP}_i^\star]_{i\in[0,P]},
S_{outp}^\star,\\
\pi_{in}^\star, \pi_{ACS}^\star,
\mathtt{ACCU}_{out}^\star,\pi_{out}^\star
\end{matrix}
\right)
\gets
\mathcal{S}\left(
\begin{matrix}
\mathtt{pp},\mathtt{pp_{ACS}},\\
\mathtt{vk},\mathtt{avk},\mathtt{dk}
\end{matrix}
\right)
\]

\item The adversary receives the simulated transcript and outputs a bit:
\[
b \gets \mathcal{A}\left(
\begin{matrix}
\mathtt{pp},\mathtt{pp_{ACS}},
\mathtt{vk},\mathtt{avk},
[\mathtt{IR}_i^\star,\mathtt{CP}_i^\star]_{i\in[0,P]},\\
\mathtt{dk},\pi_{in}^\star,\pi_{ACS}^\star,\mathtt{ACCU}_{out}^\star,
\pi_{out}^\star,S_{outp}^\star
\end{matrix}
\right)
\]

\item Output $1$ iff $b=1$.
\end{enumerate}
\end{mdframed}

\subsection{Security of IR and Predicates}

Refer to the threat model and the two properties of effective proof, we define three attacks by adversaries with capabilities in our threat model. Our following discussion shows that none of them can defeat A2-DIDM under our threat model.

\textbf{Security of continuity of weight checkpoint distance.} 
An adversary can construct a sequence of weight checkpoints that deceives verifiers (\textbf{M1}). 
We note that such a forged history of DNN models cannot simultaneously satisfy $\Phi_{\mathtt{CWCD}}$ and $\Phi_{\mathtt{IWFW}}$. 
Due to the randomness involved in $\Phi_{\mathtt{IWFW}}$, the adversary must choose a randomly initialized model and interpolate it with the stolen model. 
However, the weight distance between a randomly initialized model and a converged model is larger than the weight distance between the real initial model and the converged model. 
Specifically, the distance between the real initial model and the converged model is accumulated step by step along the same training trajectory. 
By contrast, the forged checkpoints ($\widetilde{\mathbf{W}}$) that are independently and randomly initialized do not lie on this trajectory, enabling the $\mathcal{V}$ to flag the forged checkpoint as suspicious via $\Phi_{\mathtt{CWCD}}$.

\textbf{Security of distance from initial weight to final weight.} By placing $W_{0}$ in a favorable position, $\mathcal{A}$ can steer the DNN training process towards its desired final result $f_{W_{P}}$, thereby forging a training history to cheat the verifier (\textbf{M2}). 
However, $\Phi_\mathtt{IWFW}$ ensures every weight parameter ($[W_{0}^{(t)}]_{t \in [0,T]}$) of the $T$-layer DNN model has been randomly initialized from GMM distribution, while it also guarantees the the independence between any two weight parameters of the initial model.

\textbf{Security of monotonicity of weight checkpoint distribution.} 
Since $\mathcal{A}$ owns the labeled dataset ($\mathbf{D}_\mathcal{A}$), $\mathcal{A}$ can forge weight checkpoints ($\widetilde{\mathbf{W}}$) via knowledge distillation (\textbf{M3}). 
Although distilling from a converged DNN teacher may help an adversary evade the detection of $\Phi_{\mathtt{CWCD}}$ and $\Phi_{\mathtt{IWFW}}$, the consecutive checkpoints produced by distillation usually exhibit larger distributional deviations and can still be captured by $\Phi_{\mathtt{MWCD}}$.
Specifically, knowledge distillation is optimized to fit the teacher model’s outputs, so checkpoint updates are no longer governed only by the smooth parameter updates of SGD, making the weight distributions between adjacent checkpoints more likely to exhibit abnormal fluctuations.
Hence, $\Phi_\mathtt{MWCD}$ verifies the monotonicity of model weight distributions to ensure the genuineness of the owner’s checkpoint sequence.


\textbf{Data privacy and function privacy.} 
$\mathcal{A}$ may attempt to tamper with the NIZK proof or the IR commitment, or to exploit the publicly available protocol algorithms to steal the model owner's private information (\textbf{M4}). 
A2-DIDM utilizes the unique training history to identify the DIDM, without exposing the owner's private dataset.
By the binding and hiding properties of the commitment scheme, $\mathtt{CP}$ serves as a commitment to all attributes of $\mathtt{IR}$ without revealing any information, thereby ensuring data privacy. 
Moreover, the relations among the predicates are hidden by the external SNARK prover, which guarantees function privacy. 
Therefore, the zkSNARK proofs ensures that an adversary cannot craft synthetic data points to induce a specific gradient, because when selecting data points it does not know the corresponding weight ($W_i$).

\subsection{Efficiency of Accumulator-based Audit for A2-DIDM}

\textbf{Efficiency of Accumulator-based zkSNARK Scheme.}
The efficiency requirement of A2-DIDM necessitates the accumulator \textbf{verifier} algorithm to exhibit sub-linear time complexity relative to the size of zkSNARK circuit. More precisely, the accumulator operation should be performed in sub-linear time with respect to the zkSNARK circuit and should not exhibit growth with each accumulation step\footnote{The research work in paper\cite{bowe2020zexe} has already demonstrated the efficiency of the SNARK \textbf{verifier}(the circuit shown in Figure \ref{fig3}). Therefore, the efficiency of A2-DIDM primarily depends on whether the algorithm of the accumulator \textbf{verifier} is sub-linear.}.
Our work builds upon the accumulation scheme by Benedikt et al.\cite{bunz2020proof} and Sean et al. \cite{bowe2020zexe}, and refer to their research work for detailed definitions of the accumulator \textbf{verifier} circuit.

\begin{definition}
    Let $\mathsf{AVer}_{(\lambda,n,I,M)}$ denotes the computation circuit for the accumulator \textbf{verifier}. Here, $\lambda$ represents the security parameter, $n$ represents the number of accumulator proof tuples ($\mathtt{ACCU}$,$\pi$), $I$ denotes the maximum size of the index for the binary relation $\mathbf{R}$, and $M$ denotes the maximum size of each instance. 
    Moreover, the size of $\mathsf{AVer}_{(\lambda,n,I,M)}$, $\mathtt{avk}$ and $\mathtt{ACCU}$ are denoted by $S_{\mathsf{AVer}}(\lambda,n,I,M)$, $S_{\mathtt{avk}}(\lambda,n,I)$ and $S_{\mathtt{ACCU}}(\lambda,n,I)$, respectively. 

\end{definition}

As shown in Figure \ref{fig4}, an accumulator \textbf{verifier} instance involves an accumulator verification key ($\mathtt{avk}$), an accumulator value ($\mathtt{ACCU}$), and some supplemental data of size ($k$). 
Refer to the paper\cite{bunz2020proof}, 
we assume that $p$ is the number of constraints in the predicate $\Phi$ (e.g., $\Phi_{\mathtt{CWCD}}$, $\Phi_\mathtt{IWFW}$, $\Phi_\mathtt{MWCD}$), and there exists a effective simulator $\mathcal{S}$ that implements a zero-knowledge accumulator.

We quantify the size of the accumulator \textbf{verifier} circuit as $S_{\mathsf{AVer}}^{*}(\lambda,n,I,k) $ $=$  $S_{\mathsf{AVer}}(\lambda, n, I, S_{\mathtt{avk}}(\lambda,n,I)+S_{\mathtt{ACCU}}(\lambda,n,I)+k)$. 
We note that $S_{\mathtt{ACCU}}(\lambda,n,I)$ is constrained by $\lambda$, $n$, $I$, indicating that the size of the accumulator is independent of the number of input instances. 
Moreover, we assume that for some $\mathcal{S}_{0}(n,k)=O(nk)$, the size of circuit $\mathcal{C}_{\mathsf{AVer}, \Phi}^{(\lambda,I,M)}$ is denoted as follows:
\begin{equation}\label{eq:4}
    \mathcal{S}(\lambda,p,n,k,I) =  p + \mathcal{S}_{0}(n,k) + S_{\mathsf{AVer}}^{*}(\lambda,n,I,k)
\end{equation}

In order to determine the efficiency bound of the accumulator \textbf{verifier}, our objective is to identify the minimum index size bound function, denoted as $I$, that satisfies the condition $\mathcal{S}(\lambda,p,n,k,I(\lambda,$ $p,n,k))$ $\leq$ $I(\lambda,p,n,k)$.

\begin{lemma}
    We suppose that the running time of the simulator $\mathcal{S}$, with the given parameters $\lambda$, $p$, $n$, $k$ $\in$ $\mathbb{N}$, is bounded by the size function $I(\lambda,p,n,k)$ as follows:
    $$ \mathcal{S}(\lambda,p,n,k,I(\lambda,p,n,k)) \leq I(\lambda,p,n,k).
    $$
    Specifically, we note that for any $\lambda$, $p$, $n$, $k$, the ratio of accumulator \textbf{verifier} circuit size to index size, denoted as $S_{\mathsf{AVer}}^{*}(\lambda,n,I,k)/I$, is monotonically decreasing with respect to $I$.
    Furthermore, we assume that a constant $\alpha$ $\in$ $[0,1]$ controls the rate of decay of the exponential function while an increasing polynomial $\beta$ denotes the base growth of the exponential function. For sufficiently large $\lambda,p,n,k$, two inequalities hold as follows:
    \begin{equation*}
        \begin{cases}
            S_{\mathsf{AVer}}^{*}(\lambda,n,I,k) \leq I^{1-\alpha}\beta(\lambda, n, k)\\
            I(\lambda,p,n,k) \leq O(p+\beta(\lambda,n,k)^{1/\alpha})
        \end{cases}
    \end{equation*}

\label{lemma1}
\end{lemma}

\begin{proof}
    Due to the monotone decreasing of $S_{\mathsf{AVer}}^{*}(\lambda,n,I,k)/I$, there exists a minimum integer $I_{0}$ $\gets$ $I_{0}(\lambda,n,k)$ such that $S_{\mathsf{AVer}}^{*}(\lambda,n,$ $I_{0},k)/I_{0}$ \textless $1/2$.
    We assume that $I(\lambda,p,n,k)$ takes the maximum value between $I_{0}(\lambda,n,k)$ and $2(p+\mathcal{S}_{0}(n,k))$, denoted as $I(\lambda,p,n,k)$ $\gets$ $\mathsf{max}(I_{0}(\lambda,n,k),2(p+\mathcal{S}_{0}(n,k)))$. 
    For $I$ $\gets$ $I(\lambda,p,n,k)$, the $\mathcal{S}(\lambda,$ $p,n,k,I)$ is bounded as follows:
    \begin{align*}
        \mathcal{S}(\lambda,p,n,k,I) &= p + \mathcal{S}_{0}(n,k) +I \cdot  S_{\mathsf{AVer}}^{*}(\lambda,n,I,k)/I\\
        &\textless I/2 +I/2 =I.
    \end{align*}
    Clearly $p$ $+$ $\mathcal{S}_{0}(n,k)$ = $O(p)$. 
    Furthermore, we assume that for any sufficiently large $\lambda,n,I,k$, the inequality $S_{\mathsf{AVer}}^{*}(\lambda,n,I,k) \leq I^{1-\alpha}\beta(\lambda, $ $ n, k)$ holds.
    Due to the increasing $\beta$, for $I^{*}(\lambda,n,k)$ $\gets$ $(2 \cdot \beta(\lambda,n,k))^{1/\alpha}$, it holds that
    \begin{align*}
    S_{\mathsf{AVer}}^{*}(\lambda,n,I^{*}(\lambda,n,k),k)/I^{*}(\lambda,n,k) \\
    \textless \beta(\lambda,n,k) \cdot (2 \cdot \beta(\lambda,n,k))^{-1}
    =1/2
    \end{align*}
    where the parameters $\lambda,n,k$ are sufficiently large. 
    
    Hence for all $\lambda,n,k$ sufficiently large, two inequalities $I_{0}$ $\leq$ $I^{*}= (2 \cdot \beta(\lambda,n,k))^{1/\alpha}$ and $I(\lambda,p,n,k) \leq O(p+\beta(\lambda,n,k)^{1/\alpha})$ hold. The bound $I$ of accumulator \textbf{verifier} circuit size holds.
\end{proof}
    Additionally, recall the Lemma \ref{lemma1}, for simulator $\mathcal{S}$, Eq. (\ref{eq:4}) can be further expressed as follows:
    \begin{align*}
	\mathcal{S}(\lambda,p,n,k,I) &=  p + \mathcal{S}_{0}(n,k) + S_{\mathsf{AVer}}^{*}(\lambda,n,I,k)\\
     &= p + O(I^{1-\alpha}\beta(\lambda, n, k))\\
     &= p + O(p^{1-\alpha}\beta(\lambda, n, k)+\beta(\lambda, n, k)^{1/\alpha}).
    \end{align*}
    Particularly, if $p=cr_{kp}(\beta(\lambda,n,k)^{1/\alpha})$ then the size of index is $p+o(p)$, and so the stated efficiency bounds hold.

\begin{table}[]
\caption{The costs of the pairing-based polynomial commitment scheme.}
\begin{tabular}{cc}
\hline
\begin{tabular}[c]{@{}c@{}}Costs of Accumulator Scheme for \\ Pairing-based Polynomial Commitment\end{tabular} & Performance                               \\ \hline
\textit{Accumulator Generation}    & $\mathsf{poly}(\lambda)$        \\
\textit{Accumulator Key Generetion}     & $\mathsf{poly}(\lambda)$        \\
\textit{Verify evaluation proofs}         & $\Theta(n)$ pairings                      \\
\textit{Verify an accumulation step}         & $\Theta$ $\mathbb{G}_{1}$ multiplications \\
\textit{Verify final accumulator}         & 1 pairing                                 \\
\textit{Accumulator ($\mathtt{ACCU}$) size }            & 2$\mathbb{G}_{1}$                         \\ \hline
\end{tabular}
\label{tab6}
\end{table}

\textbf{Efficiency of Accumulator-based Polynomial Commitment Scheme.}
A2-DIDM adopts a polynomial commitment scheme based on the knowledge assumption of bilinear groups. 
Table \ref{tab6} shows the costs of the accumulator-based polynomial commitment scheme in the random oracle model, where $n$ represents the number of evaluation proofs. 
Both $\mathsf{ACS.Gen}$ and $\mathsf{ACS.KeyGen}$ take a polynomial time cost $\mathsf{poly}(\lambda)$, related to the security parameter. 
Due to $P$ elements in $\mathtt{CP}$ series for $S_{inp}$ within one accumulator, the accumulator \textbf{prover} performs $O(P)$ scalar multiplications in $\mathbb{G}_{1}$. 
Thus, $\mathsf{ACS.Prove}$ executes $O(1)$ linear combinations of $O(P)$ elements in $\mathbb{G}_{1}$. 
On the other hand, the accumulator \textbf{verifier} performs the same time cost as accumulator \textbf{prover}.
Additionally, for pairing-based polynomial commitment, the accumulator ($\mathtt{ACCU}$) involves two elements in $\mathbb{G}_{1}$, thus the accumulator size performs $2\mathbb{G}_{1}$. 
For \textbf{decider}, $\mathsf{ACS.Decide}$ performs one pairing so it takes $1$ pairing time cost to verify final accumulator. 
As shown in Table \ref{tab6}, the time complexity for verifying the correctness of an accumulation step is significantly lower compared to the time complexity of verifying the evaluation proofs, while the verification of the final accumulator proof only requires one pairing operation. 
By effectively deferring the verification work of the SNARK \textbf{verifier} to the A2-DIDM \textbf{verifier} through the accumulator scheme, we avoid adding costs as the number of proofs to be verified increases.





\section{Implementation and Evaluation} \label{sec:exp}


\subsection{Experiments for A2-DIDM Proving System}

\subsubsection{Experiment Configuration}
We implement the A2-DIDM proving system with an accumulator-based commitment scheme in Rust.
Our implementation builds on the \verb|bellman-bignat|\footnote{https://github.com/alex-ozdemir/bellman-bignat} and \verb|veri-zexe|\footnote{https://github.com/EspressoSystems/veri-zexe} libraries for algebraic and finite-field operations, polynomial computations, and blockchain-oriented primitives, and uses the pairing-friendly curve BLS12-377.
In addition, \verb|bellman-bignat| provides cryptographic primitives (e.g., Pedersen, SHA-256, and Poseidon) and hash-generic Merkle trees.
Based on this stack, we construct the zkSNARK circuits and constraints in A2-DIDM.
Specifically, we employ TURBOPLONK and ULTRAPLONK for internal and external circuits, respectively. For IR commitments, A2-DIDM utilizes Poseidon hashes and Merkle accumulators for parameter generation. 
Moreover, due to the lack of native support for BLS12-377 and PLONK on mainstream public blockchains (e.g. ETH), we implement the blockchain ledger-side validation logic using \verb|veri-zexe| rather than direct contract deployment. 
Our implementation covers transaction generation, commitment construction, and transaction cryptographic checks (e.g., Merkle-root consistency and SNARK verification). All experiments are conducted locally, with results averaged over 10 times per setting.

In our evaluation, IRs are used as inputs to transactions, and we set the number of IRs to $Num \in \{2, 3,4,5,6,7,8,9,10\}$.
To emulate model-weight training, we vary the weight-vector size as $Size \in \{4,8,16,32,64,128,256\}$ and set the Merkle tree height to $h \in \{3,4,5\}$.
We observe that when $Size \in \{64,128,256\}$, the setup phase of the Pedersen-hash-based polynomial commitment takes more than 5 hours.
Given this prohibitive cost, we adjust the experimental parameters to keep the evaluation tractable.
For A2-DIDM, we fix the Merkle tree height to $h=6$ and measure runtime under $Size \in \{4,8,16,32,64,128,256\}$.
For the \textit{Pedersen commitment} baseline, we instead evaluate smaller Merkle trees by varying the height from $h=3$ to $h=5$ with $Size \in \{8,16,32,64\}$.

\begin{table*}[!t]
\caption{Experiment evaluation of transaction proof for A2-DIDM proving system}
\begin{tabular}{ccccccccccc}
\hline
\multicolumn{2}{c}{\multirow{2}{*}{A2-DIDM performance for transaction}}                                                           & \multicolumn{9}{c}{The number of IRs}                                                    \\ \cline{3-11} 
\multicolumn{2}{c}{}                                                                                                               & $Num$=2 & $Num$=3 & $Num$=4 & $Num$=5 & $Num$=6 & $Num$=7 & $Num$=8 & $Num$=9 & $Num$=10 \\ \hline
\multirow{2}{*}{\textit{\begin{tabular}[c]{@{}c@{}}Proving \\ System Setup\end{tabular}}}  & \textit{Parameter Generation(s)}      & 143.34  & 257.25  & 267.80  & 295.41  & 529.35  & 600.40  & 713.07  & 931.41  & 1373.07  \\  & \textit{Key Addresses Generation(s)}  & 0.02    & 0.02    & 0.02    & 0.02    & 0.02    & 0.02    & 0.02    & 0.02    & 0.02     \\ \hline
\multirow{3}{*}{\textit{\begin{tabular}[c]{@{}c@{}}Transaction\\ Generation\end{tabular}}} & \textit{Internal Circuit Constraints} & 53754   & 71535   & 88924   & 106313  & 124094  & 141483  & 158872  & 176653  & 194042   \\  & \textit{External Circuit Constraints} & 87176   & 126099  & 141492  & 180544  & 219468  & 234860  & 273912  & 312813  & 328205   \\  & \textit{Transaction Execute(s)}       & 196.75  & 382.14  & 399.21  & 412.16  & 763.30  & 816.59  & 844.63  & 866.57  & 894.66   \\ \hline
\textit{Verification}                                                                      & \textit{Merkle Root Verification(s)}  & 0.21    & 0.22    & 0.24    & 0.22    & 0.25    & 0.23    & 0.23    & 0.22    & 0.22     \\ \hline
\end{tabular}
\label{tab:tx}
\end{table*}

\begin{table}[!t]
\caption{Single-threaded experiment evaluation of IR commitment.
}
\begin{tabular}{ccccc}
\hline
\multicolumn{2}{c}{\multirow{2}{*}{\begin{tabular}[c]{@{}c@{}}A2-DIDM performance \\ for IR commitment\end{tabular}}}   & \multicolumn{3}{c}{Weight vector size}                      \\ \cline{3-5} 
\multicolumn{2}{c}{}                                                                                                            & \multicolumn{1}{c}{$Size$=64}     & \multicolumn{1}{c}{$Size$=128}     & $Size$=256     \\ \hline
\multirow{3}{*}{\textit{\begin{tabular}[c]{@{}c@{}}Polynomial \\ Commitment \\ Setup\end{tabular}}} & \textit{Poly. Opera.(s)}  & \multicolumn{1}{c}{192.81} & \multicolumn{1}{c}{384.194} & 766.434 \\      & \textit{Honest Query (s)} & \multicolumn{1}{c}{38.378} & \multicolumn{1}{c}{71.234}  & 141.703 \\ & \textit{Para. Gen. (s)}   & \multicolumn{1}{c}{1.351}  & \multicolumn{1}{c}{2.695}   & 5.263   \\ \hline
\multirow{4}{*}{\textit{\begin{tabular}[c]{@{}c@{}}Proof \\ Generation\end{tabular}}}               & \textit{Witness size}     & \multicolumn{1}{c}{641}    & \multicolumn{1}{c}{1281}    & 2561    \\  & \textit{Synthesize (s)}   & \multicolumn{1}{c}{4.994}  & \multicolumn{1}{c}{10.035}  & 19.814  \\      & \textbf{Prover (s)}       & \multicolumn{1}{c}{28.902} & \multicolumn{1}{c}{55.245}  & 109.399 \\ \hline
\textit{\begin{tabular}[c]{@{}c@{}}Proof \\ Validation\end{tabular}}                                & \textbf{Verifier (s)}     & \multicolumn{1}{c}{0.003}  & \multicolumn{1}{c}{0.003}   & 0.003   \\ \hline
\end{tabular}
\label{tab1}
\end{table}

\begin{table*}[]
\caption{Pedersen commitment scheme versus A2-DIDM commitment scheme.
}
\setlength{\tabcolsep}{1.8mm}{
\begin{tabular}{ccccc|ccc|ccc|ccc}
\hline
\multirow{2}{*}{Time Costs}                                                               & \multirow{2}{*}{\begin{tabular}[c]{@{}c@{}}Commitment \\ Scheme\end{tabular}} & \multicolumn{3}{c|}{Size=8} & \multicolumn{3}{c|}{Size=16} & \multicolumn{3}{c|}{Size=32} & \multicolumn{3}{c}{Size=64} \\ \cline{3-14} 
     &                                                                               & h=3     & h=4     & h=5     & h=3      & h=4     & h=5     & h=3      & h=4     & h=5     & h=3     & h=4     & h=5     \\ \hline
\multirow{2}{*}{Init time (s)}                                                            & Pedersen                                                                      & 0.013   & 0.019   & 0.021   & 0.017    & 0.025   & 0.018   & 0.033    & 0.027   & 0.029   & 0.043   & 0.047   & 0.050   \\ & Our scheme                                                                    & 0.018   & 0.018   & 0.020   & 0.031    & 0.027   & 0.039   & 0.039    & 0.054   & 0.051   & 0.069   & 0.081   & 0.102   \\ \hline
\multirow{2}{*}{\begin{tabular}[c]{@{}c@{}}Parameter \\ generation time (s)\end{tabular}} & Pedersen                                                                      & 1.061   & 1.483   & 1.634   & 2.112    & 2.979   & 3.228   & 4.191    & 5.882   & 6.437   & 8.361   & 11.868  & 13.028  \\    & Our scheme                                                                    & 0.132   & 0.153   & 0.177   & 0.258    & 0.304   & 0.347   & 0.515    & 0.620   & 0.699   & 1.018   & 1.198   & 1.373   \\ \hline
\multirow{2}{*}{Prover time (s)}                                                          & Pedersen                                                                      & 11.877  & 18.320  & 19.299  & 22.988   & 36.389  & 38.015  & 44.301   & 68.592  & 71.671  & 82.904  & 131.321 & 138.173 \\ & Our scheme                                                                    & 3.370   & 3.901   & 4.466   & 6.225    & 7.272   & 8.343   & 11.751   & 13.847  & 15.864  & 21.926  & 25.817  & 29.601  \\ \hline
\multirow{2}{*}{Verifier time (s)}                                                        & Pedersen                                                                      & 0.003   & 0.003   & 0.003   & 0.003    & 0.003   & 0.003   & 0.003    & 0.003   & 0.003   & 0.003   & 0.003   & 0.003   \\ & Our scheme                                                                    & 0.003   & 0.003   & 0.003   & 0.003    & 0.003   & 0.003   & 0.003    & 0.003   & 0.003   & 0.003   & 0.003   & 0.003   \\ \hline
\multirow{2}{*}{Synthesize time (s)}                                                      & Pedersen                                                                      & 0.318   & 0.381   & 0.441   & 0.633    & 0.758   & 0.861   & 1.271    & 1.492   & 1.730   & 2.537   & 2.9813  & 3.437   \\ & Our Scheme                                                                    & 0.438   & 0.533   & 0.652   & 0.863    & 1.063   & 1.323   & 1.718    & 2.149   & 2.587   & 3.401   & 4.295   & 5.150   \\ \hline
\end{tabular}
}
\label{tab2}
\end{table*}

\subsubsection{Experiment Results}

We evaluate the performance of the A2-DIDM proving system from two aspects: \textit{Transaction Proof} and \textit{IR Commitment}. 
On one hand, \textit{Transaction Proof} involves three main phases of A2-DIDM: \textit{Proving System Setup}, \textit{Transaction Generation} and \textit{Verification}. 
On the other hand, \textit{IR Commitment} of A2-DIDM is implemented to run in a single-threaded manner, with a focus on the time costs associated with the three main phases: \textit{Polynomial Commitment Setup}, \textit{Proof Generation}, and \textit{Proof Validation}.
In particular, we pay close attention to the time expenses incurred during the execution of these phases, as they directly impact the usability and practicality of the overall verifiable computation model. 

\textbf{Experiment Results of Accumulator-based zkSNARK.}
Table \ref{tab:tx} presents the specific time costs of \textit{Transaction Proof} for different $Num$-input transaction dimensions in A2-DIDM.
For \textit{Proving System Setup}, the model identity audit issue is transformed into a constraint system via the constructions of IR and predicates. 
\textit{Parameter Generation} involves the generation of security parameters.
\textit{Key addresses Generation} denotes the time costs to generate the public key address ($\mathtt{IRPK}$).
In addition, we evaluate the performance of the A2-DIDM proving system in the \textit{Transaction Generation} phase from three dimensions: \textit{Internal Circuit Constraints}, \textit{External Circuit Constraints}, and \textit{Transaction Proof Execution time}.
As the number of input IRs changes from $Num$=2 to $Num$=10, both the constraints on internal and external circuits, as well as the execution time of transaction proof, show a significant increase.
For \textit{Verification}, the verification time cost for Merkle tree root remains stable between 0.22s and 0.25s, achieving a lightweight verification cost.

Table \ref{tab1} would contain the time costs for different configurations of $h$ and $Size$ in IR commitment.
\textit{Poly. Opera.} and \textit{Para. Gen.} denotes the \textit{Polynomial Operation time} and \textit{Parameter Generation time}, respectively. The size of the witness includes the number of original weight vector proofs, the number of weight vector proofs, and an additional 1 (to ensure that the witness has at least one element). 
Particularly, \textit{Polynomial Operations} refer to various mathematical operations performed on polynomials within the polynomial domain to express constraints and relationships.
\textit{Honest Query} typically denotes the interaction between an honest verifier and the prover. During this interaction, the verifier sends challenges or queries to the prover to test whether the prover can correctly generate the corresponding responses. 
\textit{Parameter Generation} denotes the time cost for $\mathtt{pp_{C}}$ generation.
For \textit{Proof Generation}, we remark that the inclusion of multiple IR updates and additions would result in increased time and resource costs for IR commitment \textbf{prover}. 
Additionally, we synthesize the circuit into a constrained system and evaluate the synthesis time. 
Finally, for \textit{Proof Validation}, we evaluate the \textbf{verifier} time costs and A2-DIDM performes well on \textit{Proof Validation}. The constant validation cost demonstrates the stability and robustness of our model.






\textbf{Experiment Results of Commitment Comparison.}
We compared A2-DIDM commitment scheme to \textit{Pedersen commitment} and summarize the experiment results in Table \ref{tab2}. 
The \textit{Prover time} in commitment is primarily determined by the execution of multiple multi-exponentiations and Fast Fourier Transforms (FFTs), with the size of FFTs being linearly proportional to constraints.
Our A2-DIDM commitment scheme demonstrates significantly higher efficiency in \textbf{prover} compared to \textit{Pedersen commitment}. 
By utilizing A2-DIDM with an accumulator-based prover and Poseidon-128 hash function, the efficiency of the \textbf{prover} has an  enhancement of approximately 80\% in comparison to employing a Pedersen hash function with the same tree height $h$.
This improvement is attributed to the difference in constraint usage between the Poseidon-128 and Pedersen hash functions within the SNARK construction.
Specifically, for each hash function call inside the SNARK, the Poseidon-128 hash function only requires 316 R1CS constraints, while the Pedersen hash function requires a significantly higher number of constraints, specifically 2753 constraints per call. 
Consequently, when generating commitment proofs for a tree of the same height ($h$), A2-DIDM with the Poseidon-128 hash function achieves higher efficiency as it involves fewer constraints and computations.

\subsection{Experiments for A2-DIDM Predicates}

\subsubsection{Dataset and Models}


Our experiments are conducted on four benchmark datasets spanning three modalities: CIFAR10 and CIFAR100 for image classification, Speech Commands for audio classification, and Yahoo Answers for text classification. 
For the image classification tasks, we evaluate four typical architectures, namely LeNet, ResNet-18, AlexNet, and VGG-16, on CIFAR10 and CIFAR100. 
For the audio classification task, we adopt the medium Depthwise Separable CNN (DSCNN) on Speech Commands. 
For the text classification task, we use the Deep Pyramid CNN (DPCNN) on Yahoo Answers.
To evaluate the impact of A2-DIDM on model accuracy, we have conducted experiments using ResNet20, ResNet32, and ResNet44 on CIFAR10, and ResNet18, ResNet34, and ResNet50 on CIFAR100.
All experiments are carried out on a system equipped with an Intel(R) Xeon(R) Platinum 8352V CPU @ 2.10GHz and an RTX 5090 GPU.
We remark that our approach does not modify the model architecture or parameters; it only records intermediate variables during the training process.


\subsubsection{Adversarial Attack Settings}
Assuming $\mathcal{A}$ has white-box access to the structure and weights of $f_{W_{P}}$, and may also possess a smaller auxiliary dataset ($\mathbf{D}_\mathcal{A}$), which are split into an unlabeled subset ($\mathbf{D}_\mathcal{A}^{u}$) and a labeled subset ($\mathbf{D}_\mathcal{A}^{l}$).
To evaluate the predicate robustness against forgery attempts, we implement three adversarial attacks for three threats (\textbf{M1}--\textbf{M3}).

\textbf{Reverse Construction Attack (RCA).} $\mathcal{A}$ starts from the converged victim DNN ($f_{W_{P}}$) and fine-tunes it on an auxiliary dataset ($\mathbf{D}_{\mathcal{A}}$) while gradually increasing the label poisoning rate, thereby intentionally degrading model utility over time. 
    In addition, a parameter regularizer is introduced to pull the model weights toward a randomly initialized model ($\widetilde{W_0}$). Specifically, at $i$-th epoch, $\mathcal{A}$ minimizes
    $$
    \mathcal{L}' = \mathcal{L}_{\mathtt{CE}} + \beta \cdot \mathsf{DL}(\widetilde{W_i},\widetilde{W_0}),
    $$
    where $\mathcal{L}_{\mathtt{CE}}$ is the cross-entropy loss computed on the (partially poisoned) labeled auxiliary subset, and $\beta$ is the regularization coefficient. 

\textbf{Checkpoint Forgery Attack (CFA).} $\mathcal{A}$ fabricates a sequence of pseudo checkpoints by linearly combining the weights of a randomly initialized model and the final converged model (\(f_{W_{P}}\)). 
    Specifically, the parameters at step \(t\) are constructed as 
    \[
    \widetilde{W_i} = (1-\mu) \cdot \widetilde{W_0} + \mu \cdot W_P,
    \quad i \in [0, P],
    \]
    where \(\widetilde{W_0}\) denotes the randomly initialized weights of the $\mathcal{A}$, \(W_P\) represents the converged weights of the victim model \(f_{W_{P}}\), and \(\mu \in [0,1]\) is the interpolation coefficient \cite{liu2023provenance}. 

\textbf{ Model Distillation Attack (MDA). }
    Since \(\mathcal{A}\) does not have access to the labels of the original dataset, the knowledge distillation loss in MDA is defined as the KL divergence between the teacher model and the student model. 
    To steal the victim model ($f_{W_{P}}$) without access to its training data, \(\mathcal{A}\) trains a student model ($\widetilde{f_{W_{P}}}$) on an auxiliary dataset ($\mathbf{D}_\mathcal{A}$), which is sharing the same architecture as the victim model. 
    The student model ($\widetilde{f_{W_{P}}}$) is optimized using knowledge distillation on unlabeled samples and an additional supervised cross-entropy term on labeled samples:
    \[
    \mathcal{L}
    = \lambda_{\mathtt{KD}}\sum_{\mathclap{x\in \mathbf{D}_\mathcal{A}^{u}}}\mathcal{L}_{\mathtt{KD}}(x)
    + \sum_{\mathclap{(x,y)\in \mathbf{D}_\mathcal{A}^{l}}}
    \Big(\lambda_{\mathtt{KD}}\mathcal{L}_{\mathtt{KD}}(x)+\lambda_{\mathtt{CE}}\mathcal{L}_{\mathtt{CE}}(x,y)\Big),
    \]
    where $\lambda_{\mathtt{KD}}$ is a KL-based distillation loss, and $\lambda_\mathtt{CE}$ denotes cross-entropy.

\subsubsection{Experiment Results}
We implement the model-training part of A2-DIDM on multiple model architectures and launch RCA, CFA and MDA to assess the effectiveness and robustness of three predicates. 

\begin{table}[!t]
\caption{The comparison of $\Phi_{\mathtt{IWFW}}$, $\Phi_{\mathtt{CWCD}}$ and $\Phi_{\mathtt{MWCD}}$ after launching CFA, RCA, and MDA under different models. $\dagger$ denotes that the corresponding predicate detects the attack. For CFA and RCA, coefficients are set as $\beta = \mu = 0.05$; for MDA, regularization coefficient is $0.005$.}
\centering
\setlength{\tabcolsep}{4.0pt}
\begin{tabular}{ccccccc}
\hline
Dataset & Model & Attack & $\Phi_{\mathtt{CWCD}}$ & $\Phi_{\mathtt{IWFW}}$(1) & $\Phi_{\mathtt{IWFW}}$(2) & $\Phi_{\mathtt{MWCD}}$ \\ \hline

\multirow{16}{*}{CIFAR10}
& \multirow{4}{*}{LeNet}
& Clean & 0.121 & 0.009 & 0.262 & 0.167 \\
& & RCA   & 0.018 & 0.118$^\dagger$ & 0.495$^\dagger$ & 0.042 \\
& & CFA   & 0.125$^\dagger$ & 0.014$^\dagger$ & 0.274$^\dagger$ & 0.003 \\
& & MDA   & 0.065 & 0.016$^\dagger$ & 0.302$^\dagger$ & 0.101 \\ \cline{2-7}

& \multirow{4}{*}{ResNet18}
& Clean & 0.035 & 0.009 & 0.143 & 0.200 \\
& & RCA   & 0.003 & 0.211$^\dagger$ & 0.222$^\dagger$ & 0.015 \\
& & CFA   & 0.035$^\dagger$ & 0.015$^\dagger$ & 0.136 & 0.004 \\
& & MDA   & 0.035 &  0.012$^\dagger$ & 0.133 & 0.109 \\ \cline{2-7}

& \multirow{4}{*}{AlexNet}
& Clean & 0.028 & 0.012 & 0.128 & 0.033 \\
& & RCA   & 0.004 & 0.135$^\dagger$ & 0.146$^\dagger$ & 0.048 \\
& & CFA   & 0.032$^\dagger$ & 0.007 & 0.123 & 0.004 \\
& & MDA   & 0.052 & 0.013 & 0.131$^\dagger$ & 0.117$^\dagger$ \\ \cline{2-7}

& \multirow{4}{*}{VGG16}
& Clean & 0.025 & 0.009 & 0.137 & 0.030 \\
& & RCA   & 0.003 & 0.174$^\dagger$ & 0.142$^\dagger$ & 0.021 \\
& & CFA   & 0.028$^\dagger$ & 0.009 & 0.133 & 0.004 \\
& & MDA   & 0.027 & 0.005 & 0.133 & 0.095$^\dagger$ \\ \hline

\multirow{16}{*}{CIFAR100}
& \multirow{4}{*}{LeNet}
& Clean & 0.148 & 0.015 & 0.256 & 0.166 \\
& & RCA   & 0.021 & 0.159$^\dagger$ & 0.387$^\dagger$ &  0.054 \\
& & CFA   & 0.152$^\dagger$ & 0.009 & 0.249 &  0.005 \\
& & MDA   & 0.101 & 0.008 & 0.296$^\dagger$ & 0.100 \\ \cline{2-7}

& \multirow{4}{*}{ResNet18}
& Clean & 0.035 & 0.008 & 0.094 & 0.185 \\
& & RCA   & 0.003 & 0.217$^\dagger$ & 0.217$^\dagger$ & 0.010 \\
& & CFA   & 0.036$^\dagger$ & 0.008 & 0.099$^\dagger$ & 0.004 \\
& & MDA   & 0.035 & 0.012$^\dagger$ & 0.089 & 0.095 \\ \cline{2-7}

& \multirow{4}{*}{AlexNet}
& Clean & 0.029 & 0.011 & 0.094 & 0.034 \\
& & RCA   & 0.006 & 0.097$^\dagger$ & 0.145$^\dagger$ & 0.035 \\
& & CFA   & 0.033$^\dagger$ & 0.012 & 0.094 & 0.003 \\
& & MDA   & 0.050 & 0.009 & 0.089 & 0.120$^\dagger$ \\ \cline{2-7}

& \multirow{4}{*}{VGG16}
& Clean & 0.029 & 0.012 & 0.094 & 0.148 \\
& & RCA   & 0.003 & 0.184$^\dagger$ & 0.221$^\dagger$ & 0.018 \\
& & CFA   & 0.030$^\dagger$ &  0.010 & 0.089 & 0.004 \\
& & MDA   & 0.028 & 0.007 & 0.096$^\dagger$ & 0.133 \\ \hline

\multirow{4}{*}{\begin{tabular}[c]{@{}c@{}}Speech\\Commands\end{tabular}}
& \multirow{4}{*}{DSCNN}
& Clean & 0.233  & 0.030 & 0.198 & 0.078 \\
& & RCA   & 0.040 & 0.339$^\dagger$ & 0.492$^\dagger$ &0.044\\
& & CFA   &0.275$^\dagger$ & 0.044$^\dagger$ & 0.201$^\dagger$ & 0.003 \\
& & MDA   & 0.220  & 0.051$^\dagger$ & 0.203$^\dagger$ & 0.151$^\dagger$ \\ \hline

\multirow{4}{*}{\begin{tabular}[c]{@{}c@{}}Yahoo\\Answers\end{tabular}}
& \multirow{4}{*}{DPCNN}
& Clean & 0.026 & 0.002 & 0.149 &  0.024 \\
& & RCA   & 0.020  & 0.008$^\dagger$  & 0.146 & 0.008 \\
& & CFA   &0.159$^\dagger$ & 0.002 & 0.144  & 0.011 \\
& & MDA & 0.204  & 0.003 & 0.148 & 0.051$^\dagger$ \\
\hline

\end{tabular}
\label{tab:cfa_rca_mda}
\end{table}

\textbf{Experiment Results of RCA, CFA and MDA.}
To demonstrate the robustness of the three predicates, we conduct extensive experiments on a diverse set of tasks: four image classification models (LeNet, ResNet18, AlexNet, VGG16) evaluated on CIFAR10 and CIFAR100, an audio classification model (medium DSCNN) on Speech Commands, and a text classification model (DPCNN) on Yahoo Answers.
For each task, we compare the predicates ($\Phi_{\mathtt{IWFW}}$, $\Phi_{\mathtt{CWCD}}$ and $\Phi_{\mathtt{MWCD}}$) of clean checkpoints in A2-DIDM with those of forged checkpoints generated by RCA, CFA and MDA.
Concretely, $\Phi_{\mathtt{IWFW}}$ is decomposed into two sub-metrics: (i) $\Phi_{\mathtt{IWFW}}(1)$/$\Phi_{\mathtt{IWFW}}(\mathtt{Init \text{ vs } GMM})$, measured as the distance between the empirical distribution of the initial parameters and the GMM distribution; and (ii) $\Phi_{\mathtt{IWFW}}(2)$/$\Phi_{\mathtt{IWFW}}(\mathtt{PCA})$, quantified as the maximum explained variance ratio obtained by PCA within each parameter group of the initial model. 
The $\Phi_{\mathtt{IWFW}}$ is satisfied only when both the $\Phi_{\mathtt{IWFW}}(1)$ and $\Phi_{\mathtt{IWFW}}(2)$ meet their respective thresholds. 
$\Phi_{\mathtt{CWCD}}$ is computed as the weight distance between the initial and final weight checkpoints.
Specifically, $\Phi_{\mathtt{CWCD}}$ is satisfied only if
$\Phi_{\mathtt{CWCD}} < $ $ \mathtt{DisMean}- 5 \times \mathtt{DisStd}$ ($\epsilon=5$).
The $\Phi_{\mathtt{MWCD}}$ is quantified as the maximum value of distribution distance among the adjacent checkpoint pairs.
In addition, for both RCA and MDA, we assume that the auxiliary dataset ($\mathbf{D}_{\mathcal{A}}$) is drawn from the same underlying distribution as the private dataset ($\mathbf{D}_{\mathcal{P}}$), with $|\mathbf{D}_{\mathcal{A}}| = 0.2\,|\mathbf{D}_{\mathcal{P}}|$. 
Moreover, we assume that $50\%$ of $\mathbf{D}_{\mathcal{A}}$ is labeled; i.e., the labeled subset ($\mathbf{D}_{\mathcal{A}}^{l}$) satisfies $|\mathbf{D}_{\mathcal{A}}^{l}| = 0.5\,|\mathbf{D}_{\mathcal{A}}|$.

Table~\ref{tab:cfa_rca_mda} shows the values of $\Phi_{\mathtt{IWFW}}$, $\Phi_{\mathtt{CWCD}}$ and $\Phi_{\mathtt{MWCD}}$ under three adversarial attacks.
RCA attempts to reconstruct a plausible training history by reversing a deliberately degraded fine-tuning process and regularizing weights toward a random model. 
Across four datasets, $\Phi_{\mathtt{IWFW}}$ consistently identifies RCA, indicating that the forged initial checkpoint violates the expected randomness and independence properties of genuine initialization. 
In contrast, $\Phi_{\mathtt{CWCD}}$ is largely insensitive to RCA, as the optimization still follows a continuous SGD evolution without producing abnormal distance peaks. 
Moreover, CFA fabricates intermediate checkpoints via linear interpolation between a random initialization and the converged victim weights, bypassing real training dynamics. 
In this attack, $\Phi_{\mathtt{CWCD}}$ provides the strongest and most stable detection, reflecting that the forged checkpoints do not satisfy the contraction behavior expected from SGD convergence.
MDA trains a student model to mimic the victim using distillation objectives on auxiliary dataset, resulting in checkpoints that differs from the victim’s true optimization path. 
Thus, $\Phi_{\mathtt{MWCD}}$ maintains highly effective against MDA, since its produced weight checkpoints fail to exhibit the expected smooth parameter updates.
$\Phi_{\mathtt{IWFW}}$ is less reliable for MDA, since distillation typically begins with a valid random initialization and may partially satisfy initialization-based constraints.
As shown in Table~\ref{tab:cfa_rca_mda}, none of the three adversarial attacks (RCA, CFA, and MDA) can satisfy all three predicates simultaneously. 
These experiment results indicate that the predicates ($\Phi_{\mathtt{IWFW}}$, $\Phi_{\mathtt{CWCD}}$ and $\Phi_{\mathtt{MWCD}}$) in A2-DIDM effectively capture the training integrity of DIDM and are robust against these adversarial attacks.


\begin{figure*}[!t]
    \centering
    \subfloat{
        \centering
	\includegraphics[width=0.5\linewidth]{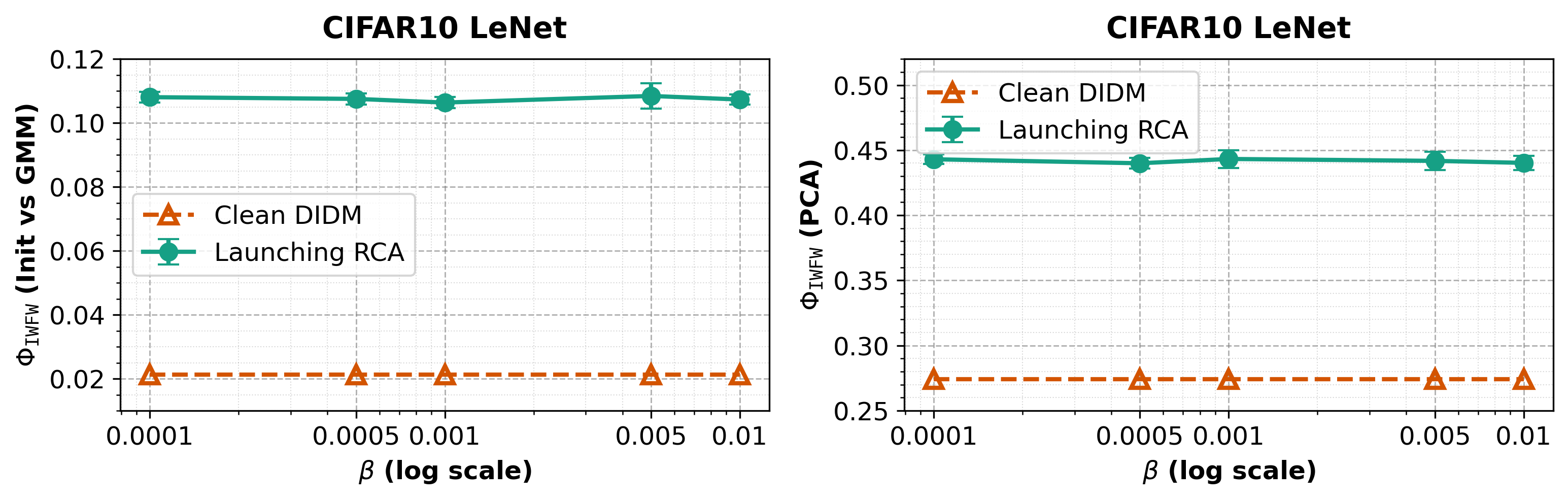}
        \label{fig:p6-le}
    }
    \subfloat{
        \centering
	\includegraphics[width=0.5\linewidth]{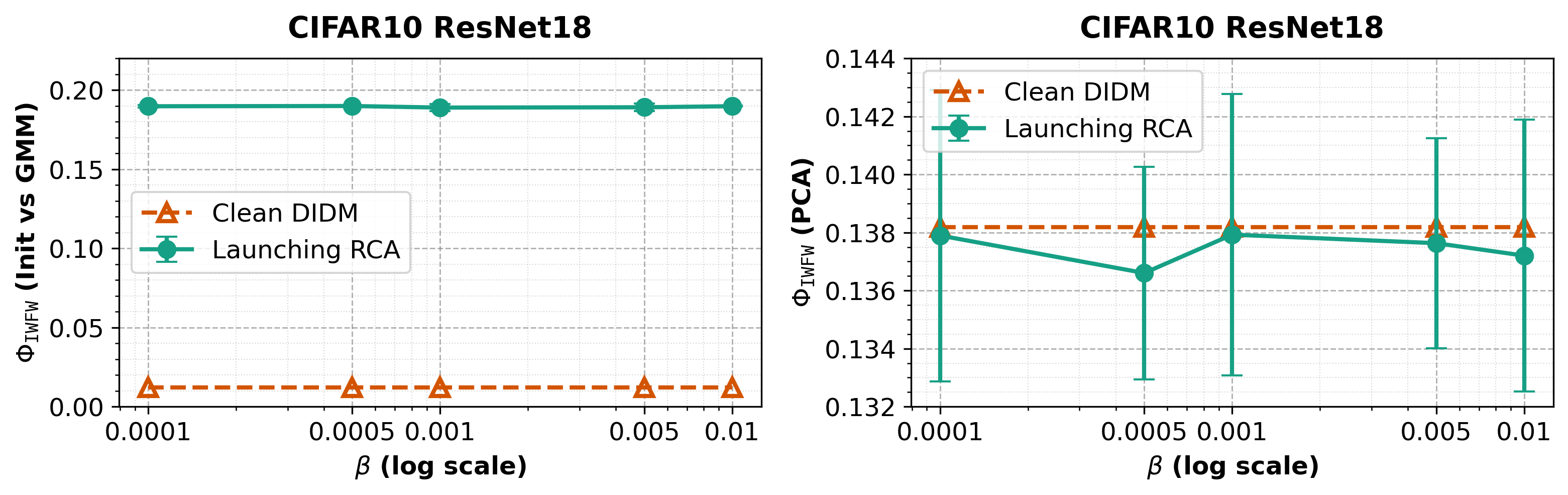}
        \label{fig:p6-res}
    }

    \subfloat{
        \centering
	\includegraphics[width=0.5\linewidth]{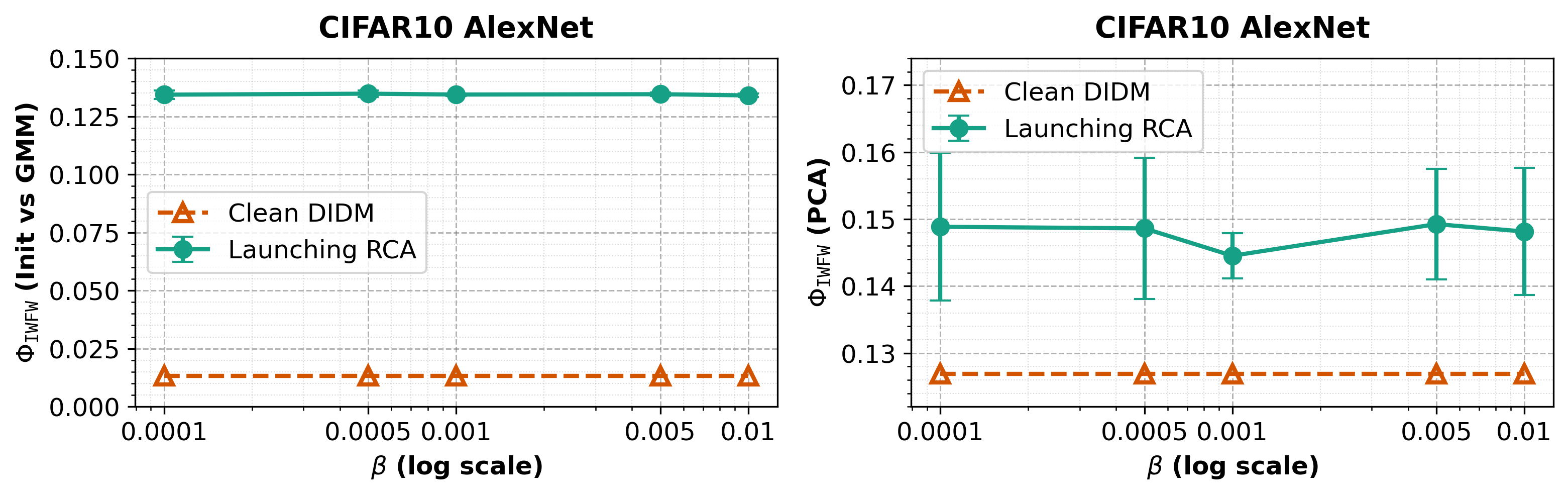}
        \label{fig:p6-alex}
    }
        \subfloat{
        \centering
	\includegraphics[width=0.5\linewidth]{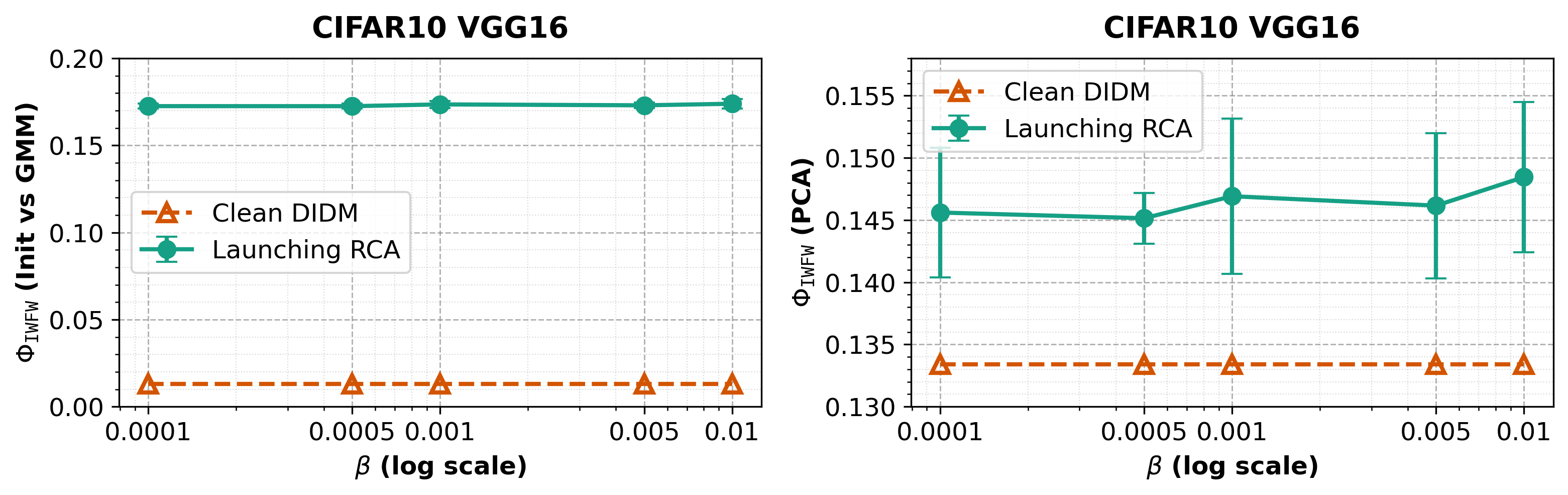}
        \label{fig:p6-vgg}
    }
    
    \subfloat{
        \centering
	\includegraphics[width=0.5\linewidth]{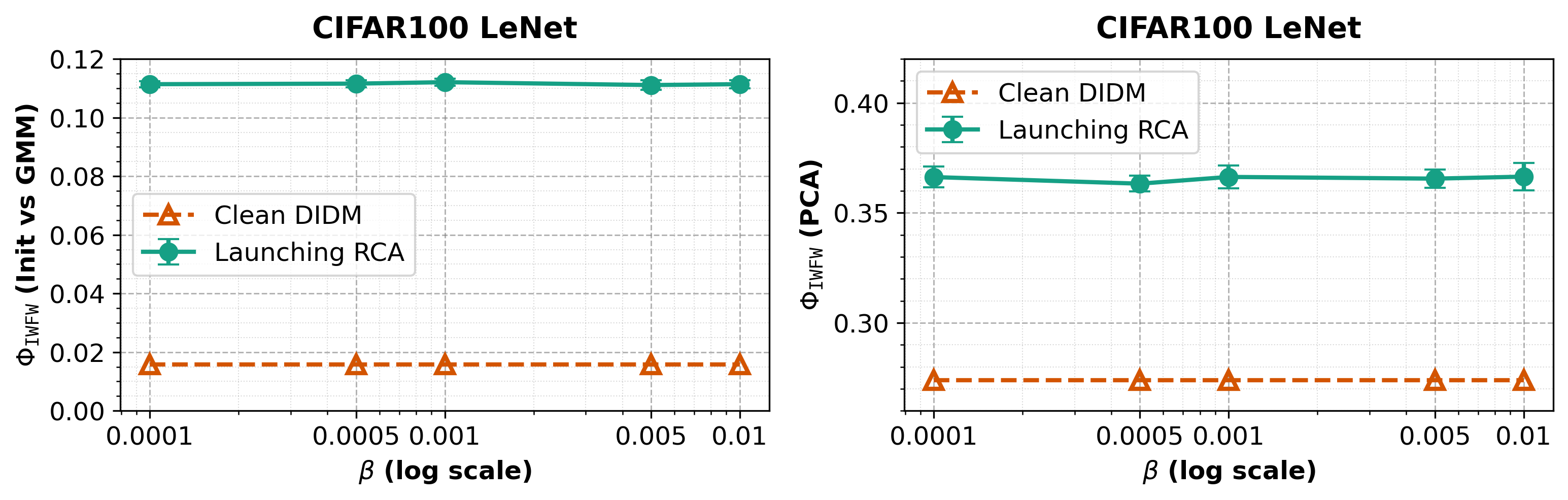}
        \label{fig:c100-alex}
    }
        \subfloat{
        \centering
	\includegraphics[width=0.5\linewidth]{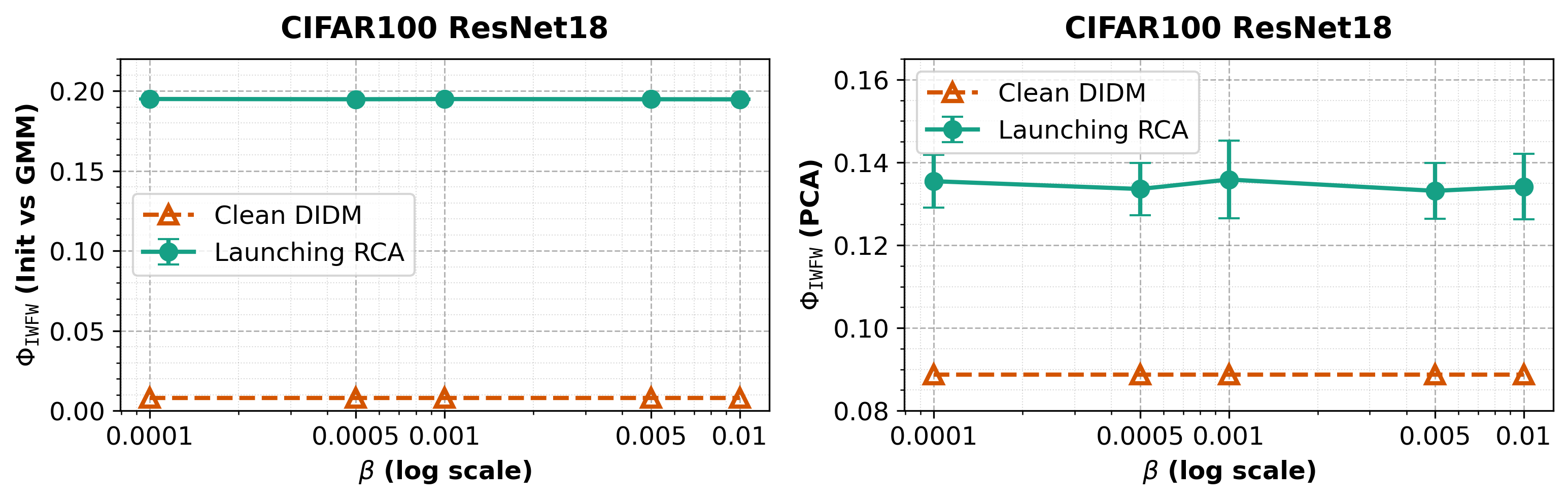}
        \label{fig:c100-vgg}
    }
    
    \subfloat{
   \includegraphics[width=0.5\linewidth]{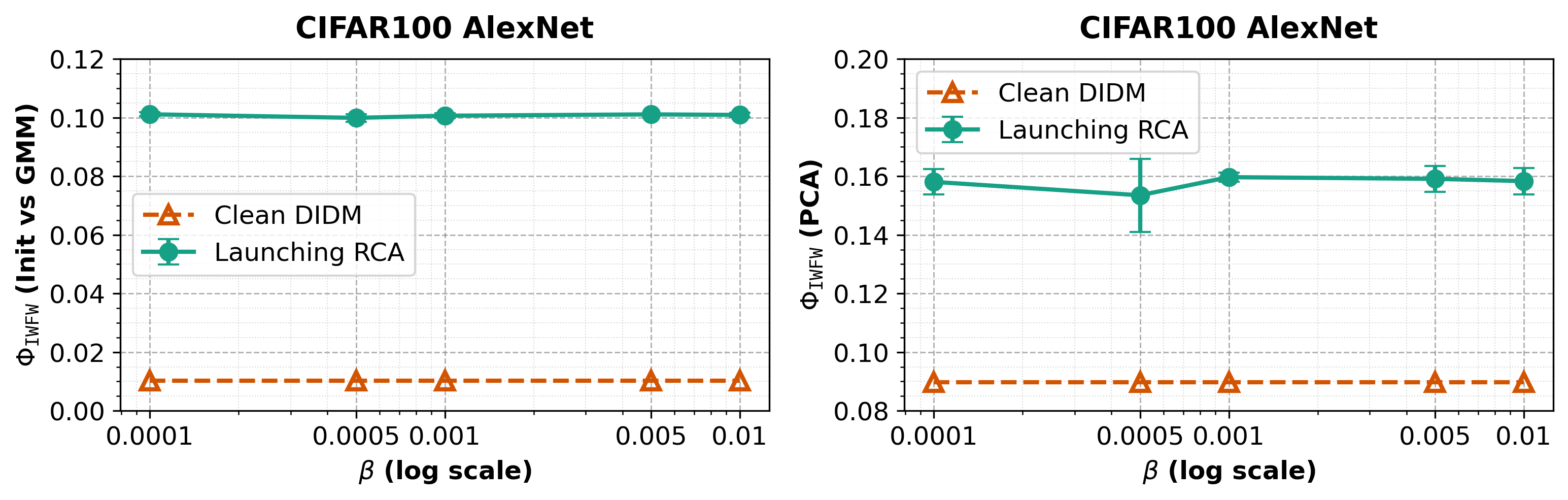}
        \label{fig:c100-alex}
    }
        \subfloat{
  \includegraphics[width=0.5\linewidth]{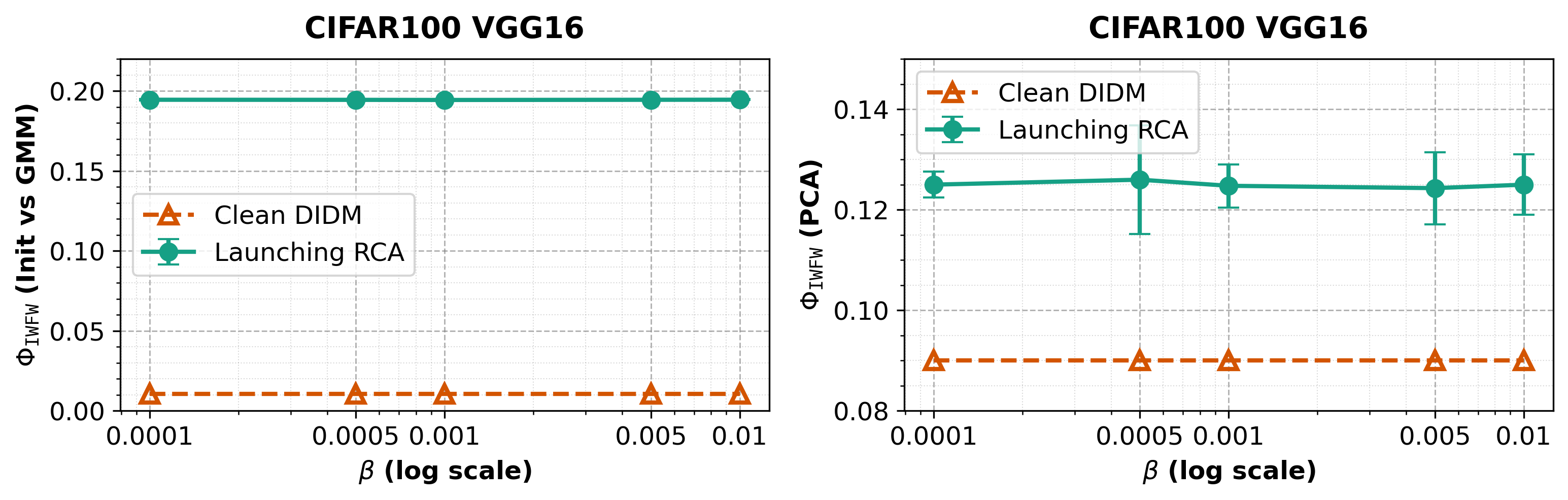}
        \label{fig:c100-vgg}
    }

\caption{The comparison of $\Phi_{\mathtt{IWFW}}(\mathtt{Init \text{ vs } GMM})$ and $\Phi_{\mathtt{IWFW}}(\mathtt{PCA})$ after launching RCA with different values of $\beta$. “Clean DIDM” denotes the original and unattacked weight checkpoint.
}
\label{figp6}
\label{fig:predicate}
\end{figure*}

Moreover, to further demonstrate the robustness of $\Phi_{\mathtt{IWFW}}$, we evaluate how launching RCA with different regularization coefficients ($\beta$) affects the initial distribution distance and PCA. 
Fig. \ref{fig:predicate} shows the values of $\Phi_{\mathtt{IWFW}}(\mathtt{Init\text{ vs }GMM})$ and $\Phi_{\mathtt{IWFW}}(\mathtt{PCA})$ after launching RCA under different regularization coefficients $\beta$ across CIFAR10/100 and four architectures. 
For each regularization coefficients ($\beta$), we conduct 5 independent RCA runs. 
RCA starts from a converged model and reversely constructs the checkpoint sequence, which allows it to bypass $\Phi_{\mathtt{CWCD}}$ detection. 
However, the forged initial weights still fail to match the distributional consistency and parameter independence expected from a genuine random initialization. 
Moreover, as shown in Fig. \ref{fig:predicate}, under different values of $\beta$, $\Phi_{\mathtt{IWFW}}$ after launching RCA remains significantly higher than that of clean DIDM. This indicates that the adversary can hardly evade the initialization verification enforced by $\Phi_{\mathtt{IWFW}}$ via tuning the regularization strength.
An additional observation is that ResNet18, equipped with residual connections and batch normalization, can quickly learn strong feature reuse on CIFAR10, making $\Phi_{\mathtt{IWFW}}(\mathtt{PCA})$ less discriminative in distinguishing clean DIDM from RCA-forged checkpoint chains. 
Nevertheless, $\Phi_{\mathtt{IWFW}}(\mathtt{Init\ \text{vs}\ GMM})$ remains consistently effective under the same setting. Therefore, the overall predicate $\Phi_{\mathtt{IWFW}}$ still provides robust detection against RCA.


\begin{table}[!t]
\caption{Validation accuracy of different DNN model.}
\label{tab6}
\begin{tabular}{ccccc}
\hline
\multirow{2}{*}{Dataset}  & \multirow{2}{*}{Model} & \multicolumn{2}{c}{Validation Accuracy} & \multirow{2}{*}{\begin{tabular}[c]{@{}c@{}}A2-DIDM \\ vs Original\end{tabular}}  \\ \cline{3-4}
                          &                        & Original           & A2-DIDM            &                                   \\ \hline
\multirow{4}{*}{CIFAR10}  & LeNet                  & 77.20\%           & 77.53\%           & +0.33\%                           \\
                          & ResNet18               & 95.63\%           & 95.26\%           & -0.37\%                           \\
                          & AlexNet                & 92.31\%           & 91.52\%           & -0.79\%                           \\
                          & VGG16                  & 93.97\%           & 93.02\%           & -0.95\%                           \\ \hline
\multirow{4}{*}{CIFAR100} & LeNet                  & 43.80\%            & 44.46\%           & +0.66\%                           \\
                          & ResNet18               & 75.97\%            & 74.53\%           & -1.44\%                          \\
                          & AlexNet                & 70.71\%           & 69.56\%           & -1.15\%                           \\
                          & VGG16                  & 74.14\%           & 74.44\%           & +0.30\%                           \\ \hline
Speech Commands           & DSCNN                  & 94.77\%           & 94.28\%           & -0.49\%                           \\ \hline
Yahoo Answers             & DPCNN                  & 72.25\%           & 70.97\%           & -1.28\%                           \\ \hline
\end{tabular}
\end{table}

\begin{table}[!t]
\centering
\caption{Model accuracy comparison of different model ownership verification schemes.}
\setlength{\tabcolsep}{3.2pt}
\begin{tabular}{ccccccc}
\toprule
\multirow{2}{*}{Dataset} & \multirow{2}{*}{Model} & \multicolumn{5}{c}{Schemes} \\
\cmidrule(lr){3-7}
& & Original & A2-DIDM & \cite{jia2021proof} & \cite{tyagi2025deepverifier} & \cite{tondi2024robust} \\
\midrule
\multirow{3}{*}{CIFAR10}
& ResNet20 & 91.73\% & 91.77\% & 89.54\% & 79.14\% & 91.30\% \\
& ResNet32 & 92.97\% & 92.67\% & 90.02\% & 80.75\% & 93.02\% \\
& ResNet44 & 94.16\% & 92.56\% & 90.46\% & 86.58\% & 93.22\% \\
\midrule
\multirow{3}{*}{CIFAR100}
& ResNet18 & 75.97\% & 74.53\% & 74.47\% & 64.58\% & 74.45\% \\
& ResNet34 & 76.85\% & 74.93\% & 74.87\% & 68.48\% & 75.38\% \\
& ResNet50 & 78.55\% & 76.96\% & 76.94\% & 70.51\% & 76.27\% \\
\bottomrule
\end{tabular}
\label{tab4}
\end{table}

\textbf{Experiment Results of Model Accuracy.} 
To evaluate the impact of A2-DIDM on model accuracy, we have implemented A2-DIDM-based training across multiple model architectures and compared the validation accuracy between A2-DIDM training and original DNN training. 
The experiment results are presented in Table \ref{tab6}. 
As shown in Table \ref{tab6}, under the same dataset and training task, A2-DIDM training achieves validation accuracy comparable to that of the original DNN training.
Since A2-DIDM only records weight checkpoints during training without altering the model architecture, its effect on DNN model performance is negligible.

In addition, we have conducted comparative experiments between A2-DIDM, PoL \cite{jia2021proof}, DeepVerifier \cite{tyagi2025deepverifier}, WBMB \cite{tondi2024robust}, and the original DNN models. The experimental results are summarized in Table \ref{tab4}. Among the methods, DeepVerifier \cite{tyagi2025deepverifier} is a hybrid watermarking approach combining black-box and white-box models, WBMB \cite{tondi2024robust} employs white-box multi-bit watermarking, and PoL \cite{jia2021proof} is based on replay training for model ownership protection. 
We compared the impact of different model ownership protection techniques on the test accuracy across two datasets and six ResNet models. 
Both DeepVerifier \cite{tyagi2025deepverifier} and WBMB \cite{tondi2024robust} embed 256-bit watermarks. 
The experimental results demonstrate that, compared to PoL \cite{jia2021proof} and DeepVerifier \cite{tyagi2025deepverifier}, the A2-DIDM method consistently achieves higher accuracy across different models. 
This suggests that A2-DIDM effectively preserves model performance while ensuring both data and functional privacy. 
Furthermore, under the same dataset and training task, A2-DIDM exhibits accuracy levels similar to those of the original DNN models and the more robust WBMB \cite{tondi2024robust}, implying that our method has negligible impact on DNN model performance.


\begin{table}[!t]
\centering
\caption{TPR and FPR for different predicate groups.}
\label{tab:ablation}
\setlength{\tabcolsep}{5pt}
\begin{tabular}{cccccccc}
\hline
Metrics & $\Phi_{\mathtt{CWCD}}$ & $\Phi_{\mathtt{IWFW}}$ & $\Phi_{\mathtt{MWCD}}$ & \begin{tabular}[c]{@{}c@{}}$\Phi_{\mathtt{CWCD}}$+\\  $\Phi_{\mathtt{IWFW}}$\end{tabular} & \begin{tabular}[c]{@{}c@{}}$\Phi_{\mathtt{MWCD}} $+\\ $ \Phi_{\mathtt{CWCD}}$\end{tabular} & \begin{tabular}[c]{@{}c@{}}$\Phi_{\mathtt{IWFW}} $+ \\ $\Phi_{\mathtt{MWCD}}$\end{tabular} & All   \\ \hline
TPR     & 1.000      & 1.000        & 1.000      & 1.000     & 1.000             & 1.000     & 1.000 \\
FPR     & 0.669       & 0.292    & 0.660      & 0.198         & 0.339   & 0.094  & 0.000 \\ \hline
\end{tabular}
\end{table}

\textbf{Experiment Results of Ablation.}
We conduct an ablation study on three predicate groups. Specifically, we launch RCA, CPA, and MDA attacks on six models; for each attack, we run 10 trials and generate 10 adversarial samples. We report the True Positive Rate (TPR) as the fraction of clean DIDMs that pass the predicate checks, and the False Positive Rate (FPR) as the fraction of attack DIDMs that pass the checks. 
Table \ref{tab:ablation} summarizes the results for different predicate combinations, including individual predicates, pairwise combinations, and all predicates. 
The results show that all configurations achieve a TPR of 1.000 on clean DIDMs, indicating no loss on benign samples. 
Moreover, using any single predicate yields a higher FPR on attack DIDMs, while combining predicates substantially reduces false positives.
By combining two predicates, the FPR can be further reduced, which demonstrates the complementarity among different predicate groups. 
When all three predicates are jointly employed, the FPR is reduced to 0.000, demonstrating that the full predicate set provides the most robust detection capability against adversarial DIDMs.


\section{Conclusion} \label{sec:con}
The issue of protecting model intellectual property arising from the commercialization of models has made model ownership auditing an important component.
We have introduced accumulators and the SNARK prooving system into DIDM auditing, and proposed the privacy-preserving A2-DIDM. 
By leveraging blockchain technology, the model owner can perform DNN training computations off-chain, while on-chain lightweight auditing operations are employed to verify the computational integrity and correctness, without revealing any specific information of DNN training.
Additionally, we perform experiments to assess the effectiveness and usability of A2-DIDM for privacy protection.

\bibliographystyle{IEEEtran}
\bibliography{sample-base}

\vfill

\end{document}